%% file: paper.tex
\documentclass[11pt]{article}
\usepackage[ruled,vlined]{algorithm2e}
\usepackage{verbatim}
\usepackage{amsmath}
\usepackage{amssymb}
\usepackage{amsthm}
\usepackage{authblk} 
\usepackage{lineno}
\usepackage{fullpage}
\usepackage{color}
\usepackage{url}
\usepackage{times}
\usepackage{graphicx}
\usepackage{epsfig}
\usepackage{epstopdf}

\newtheorem{theorem}{Theorem}
\newtheorem{lemma}[theorem]{Lemma}
\newtheorem{claim}[theorem]{Claim}
\newtheorem{corollary}[theorem]{Corollary}

\newtheorem{definition}{Definition}

\newcommand{\XOR}{\oplus}

\newcommand{\eps}{\epsilon}

\newcommand{\abs}[1]{\left| #1 \right|}

\newcommand{\T}{\mathcal{T}}
\newcommand{\I}{\mathcal{I}}
\newcommand{\A}{\mathcal{A}}
\renewcommand{\O}{\mathcal{O}}
\newcommand{\M}{\mathcal{M}}

\newcommand{\W}{\mathcal{W}}
\newcommand{\C}{\mathcal{C}}
\newcommand{\D}{\mathcal{D}}
\renewcommand{\S}{\mathcal{S}}
\newcommand{\calE}{\mathcal{E}}

\newcommand{\poly}{\mathsf{poly}}
\newcommand{\pre}{\mathtt{pre}}
\newcommand{\str}{\mathtt{str}}
\newcommand{\LCP}{{LCP}}
\newcommand{\ST}{\mathtt{ST}}
\newcommand{\PE}{\mathtt{PE}}

\title{Edit Distance: Sketching, Streaming and Document Exchange\footnote{Full  
version of an article to be presented at the 57th Annual IEEE Symposium on 
Foundations of Computer Science (FOCS 2016).}}

\author{Djamal Belazzougui \thanks{DTISI, CERIST Research Center, Algiers, Algeria. {dbelazzougui@cerist.dz}} \quad \quad \quad \quad \quad Qin Zhang \thanks{Indiana University Bloomington, Bloomington, IN, United States. {qzhangcs@indiana.edu}.  Work supported in part by NSF CCF-1525024, and IU’s
Office of the Vice Provost for Research through the FRSP.}}

\begin{document}

\begin{titlepage}
\maketitle

\begin{abstract}
We show that in the document exchange problem, where Alice holds $x \in \{0,1\}^n$ and Bob holds $y \in \{0,1\}^n$, Alice can send Bob a message of size $O(K(\log^2 K+\log n))$ bits such that Bob can recover $x$ using the message and his input $y$ if the edit distance between $x$ and $y$ is no more than $K$, and output ``error'' otherwise.  Both the encoding and decoding can be done in time $\tilde{O}(n+\poly(K))$.  This result significantly improves the previous communication bounds under polynomial encoding/decoding time.  We also show that in the referee model, where Alice and Bob hold $x$ and $y$ respectively,
they can compute sketches of $x$ and $y$ of sizes  $\poly(K \log n)$ bits (the encoding), and send to the referee, who can then compute the edit distance between $x$ and $y$ together with all the edit operations if the edit distance is no more than $K$, and output ``error'' otherwise (the decoding).  To the best of our knowledge, this is the {\em first} result for sketching edit distance using $\poly(K \log n)$ bits.  Moreover, the encoding phase of our sketching algorithm can be performed by scanning the input string in one pass. Thus our sketching algorithm also implies the {\em first} streaming algorithm for computing edit distance and all the edits exactly using $\poly(K \log n)$ bits of space.  
\end{abstract}
\end{titlepage}

\input{intro}

\input{preliminaries}

\input{doc-exchange}

\input{sketching}

\input{streaming}

\input{related}

\bibliographystyle{plain}
\bibliography{paper}

\appendix
\input{appendix}

\end{document}

%% file: intro.tex
\section{Introduction}
In this paper we study the problem of {\em edit distance}, where given two strings $s,t \in \{0,1\}^n$, we want to compute $ed(s,t)$, which is defined to be the minimum number of insertions, deletions and substitutions to convert $s$ to $t$.  We also want to find all the edit operations.  This problem has been studied extensively in the literature due to its numerous applications in bioinformatics (comparing the similarity of DNA sequences), natural language processing (automatic spelling correction), information retrieval, etc.  In this paper we are interested in the small distance regime, that is, given a threshold $K$, we want to output $ed(s,t)$ together with all the edit operations if $ed(s,t) \le K$, and output ``error'' otherwise.  We will explain shortly that this is the interesting regime for many applications.   We consider three different settings:
\begin{itemize}
\item {\em Document Exchange}.   We have two parties Alice and Bob, where Alice holds $s$ and Bob holds $t$. The task is for Alice to compute a message $msg$ based on $s$ (the encoding) and send to Bob, and then Bob needs to recover $s$ using $msg$ and his input $t$ (the decoding).  We want to minimize both the message size and the encoding/decoding time.

\item  {\em Sketching}.  We have Alice and Bob who hold $s$ and $t$ respectively, and a third party called the {\em referee}, who has no input. The task is for Alice and Bob to compute sketches $sk(s)$ and $sk(t)$ based on their inputs $s$ and $t$ respectively (the encoding), and send them to the referee.  The referee then computes $ed(s,t)$ and all the edits using $sk(s)$ and $sk(t)$ (the decoding).  The goal is to minimize both the sketch size and the encoding/decoding time.

\item  {\em Streaming}.   We are allowed to scan string $s$ from left to right once, and then string $t$ from left to right once, using a memory of small size. After that we need to compute $ed(s,t)$ and all the edits using the information retained in the memory.  The goal is to minimize the memory space usage and the processing time.
\end{itemize}
\vspace{-5mm}

\paragraph{Motivations.}
Document exchange is a classical problem that has been studied for decades.  This problem finds many applications, for example, two versions of the same file are stored in two different machines and need to be synchronized, or we want to restore a file transmitted through a channel with erroneous insertions, deletions and substitutions.  It is useful to focus on the small distance regime since one would expect that in the first application the two files will not differ by much, and in the second application the channel will not introduce too many errors. Otherwise we can detect the exception and ask the sender to transmit the whole string again, which is a low probability event thus we can afford.

Sketching the edit distance is harder than document exchange because in the decoding phase the referee does not have access to any of the original strings, which, however, also makes the problem more interesting and useful. For example, in the {\em string similarity join}, which is a fundamental problem in databases,\footnote{See, for example, a competition on string similarity join held in conjunction with EDBT/ICDT 2013: \url{http://www2.informatik.hu-berlin.de/~wandelt/searchjoincompetition2013/}} one needs to find all pairs of strings (e.g., genome sequences) in a database that are close (no more than a given threshold $K$) with respect to edit distance.  This is a computationally expensive task. With efficient sketching algorithms we can preprocess each string to a small size sketch, and then compute the edit distances on those small sketches without losing any accuracy (all of our algorithms aim at exact computations). Note that this preprocessing step can be fully parallelized in modern distributed database systems such as MapReduce and Spark. Moreover, we will show that the encoding phase can be done by scanning the input string once in the streaming fashion, and is thus time and space efficient.

\vspace{-2mm}

\paragraph{Our Results.}
In this paper we push the frontiers further for all of the three problems.  For the convenience of the presentation we use $\tilde{O}(f)$ to denote $f \poly(\log f)$, and further assume that $K \le n^{1/c}$ for a sufficiently large constant $c > 0$ (thus $n$ absorbs $\poly(K)$ factors).  We list our results together with previous results in Table~\ref{tab:results}.  Our contribution includes:
\begin{enumerate}
\item  We have improved the communication cost of the document-exchange problem to $O(K(\log^2 K + \log n))$ bits while maintaining almost linear running time.  Note that in the case when $\log K = O(\sqrt{\log n})$, our communication bound matches the information theoretic lower bound $\Omega(K \log n)$.

\item  We have obtained a sketch of size $O(K^8 \log^5 n)$ bits for edit distance, which,  to the best of our knowledge, is the {\em first} result for sketching edit distance in size $\poly(K \log n)$. This result answers an open problem in \cite{Jo12}, and is in contrast to the lower bound result in \cite{AGMP13} which states that any {\em linear} sketch for edit distance has to be of size $\Omega(n/\alpha)$ even when we want to distinguish $ed(s,t) \ge 2\alpha$ or $ed(s,t) \le 2$.

\item Our sketching algorithm can be implemented in the standard streaming model.  To the best of our knowledge, this is the {\em first} efficient streaming algorithm for exact edit distance computation.  We also show that if we are allowed to scan $s$ and $t$ simultaneously in the coordinated fashion, we can compute the edit distance using only $O(K \log n)$ bits of space, which significantly improves the result in \cite{CGK16}.
\end{enumerate}

\begin{table*}[t]
  \centering
  \begin{tabular}{|c|c|c|c|c|}
	\hline
	problem & comm. / size / space (bits) & running time & rand. or det. & ref. \\
	\hline
	document-  & $O(K\log n)$ & $n^{O(K)}$ & D & \cite{Orlitsky91}   \\
	exchange & $O(K \log(n/K) \log n)$ &  $\tilde{O}(n)$ & R & \cite{IMS05} \\
	& $O(K \log^2 n \log^* n)$ & $\tilde{O}(n)$ & R & \cite{Jo12} \\
	& $O(K^2 + K \log^2 n)$ & $\tilde{O}(n)$ & D & \cite{belazzougui2015} \\
	& $O(K^2 \log n)$ & $\tilde{O}(n)$ & R & \cite{CGK16} \\
	& $O(K(\log^2 K + \log n))$ & $\tilde{O}(n)$ & R & \bf new \\
	\hline
	sketching & $O(K^8 \log^5 n)$ & $\tilde{O}(K^2 n)$ (enc.),  & R & \bf new\\
	&&$\poly(K \log n)$ (dec.) && \\
	\hline
       streaming & $O(K^8 \log^5 n)$ & $\tilde{O}(K^2 n)$ & R & \bf new\\
	\hline
	simultaneous-  & $O(K^6 \log n)$ & $\tilde{O}(n)$ & R & \cite{CGK16} \\
	
	 streaming & $O(K \log n)$ & $O(n)$ & D & \bf new\\
	\hline
  \end{tabular}
  \caption{Our results for computing edit distance in different models.  $n$ is the input size and $K$ is a given upper bound  of the edit distance. We have assumed that $K \le n^{1/c}$ for a sufficiently large constant $c > 0$, and thus $n$ absorbs $\poly(K)$ factors.
R stands for {\em randomized} and  D stands for {\em deterministic}. }  
  \label{tab:results}
\end{table*}


\vspace{-4mm}
\paragraph{Related Work}
We survey the previous work in the settings that we consider in this paper. 
\smallskip

\noindent{\em Document Exchange.}\ \    The first algorithm for document exchange was proposed by Orlitsky~\cite{Orlitsky91}. The idea is that using graph coloring Alice can (deterministically) send Bob a message of size $O(K \log n)$, and then Bob can recover $x$ exactly using the received message and his input $y$.  However, the decoding procedure requires time exponential in $K$.   Alternatively, Alice and Bob can agree on a random hash function $h : \{0,1\}^n \to [c_K K \log n]\ (c_K = O(1))$; Alice simply sends $h(x)$ to Bob, and then Bob enumerates all vectors in the set $\{z\ |\ ed(y, z) \le K, z \in \{0,1\}^n \}$, and tries to find a $z$ such that $h(z) = h(x)$.  This protocol can succeed with high probability by choosing the constant $c_K$ large enough, but again the decoding time is exponential in $K$.  Orlitsky left the following question: 
{\em Can we design a communication {\em and} time efficient protocol for document exchange?}

Progress has been made since then~\cite{CPSV00,IMS05,Jo12}. Irmak et al.~\cite{IMS05} proposed a randomized protocol using the erasure-code that achieves $O(K \log(n/K) \log n)$ bits of communication and $\tilde{O}(n)$ encoding/decoding time, and (independently) Jowhari~\cite{Jo12} gave a randomized protocol using the ESP-tree~\cite{cormode2007string} that achieves $O(K \log^2 n \log^* n)$ bits of communication and $\tilde{O}(n)$ encoding/decoding time.  Very recently Chakraborty et al.~\cite{CGK16} obtained a protocol with  $O(K^2 \log n)$ bits of communication and $\tilde{O}(n)$ encoding/decoding time, by first embedding strings in the edit space to the Hamming space using random walks, and then perform the document exchange in the Hamming space.  We note that random walk has also been used for computing the Dyck language edit distance in an earlier paper by Saha~\cite{saha2014dyck}.

The document exchange problem is closely related to the theory of error correcting code: a deterministic protocol for document exchange naturally leads to an error correcting code.  The first efficient deterministic protocol for document exchange has been proposed very recently~\cite{belazzougui2015}; it uses $O(K^2 + K \log^2 n)$ bits of communication and runs in $\tilde{O}(n)$ time, and thus immediately gives an efficient error correcting code with redundancy $O(K^2 + K \log^2 n)$.  Also very recently, Brakensiek et al.~\cite{brakensiek2015efficient} showed an efficient error correcting code over a channel of at most $K$ insertions and deletions with a redundancy of $O(K^2 \log K \log n)$.  We note that our protocol is randomized. It remains an interesting open problem whether we can derandomize our protocol and obtain a better error correcting code.
\smallskip

\noindent{\em Sketching.}\ \    While the approximate version has been studied extensively in the literature \cite{BYJKK04,CK06,OR07,AIK09}, little work has been done for sketching edit distance without losing any accuracy.  Jowhari~\cite{Jo12} gave a sketch of size $\tilde{O}(K\log^2 n)$ bits for a special case of edit distance called the Ulam distance, where the alphabet size is $n$ and each string has no character repetitions.\footnote{In this paper when considering edit distance we always assume binary alphabet, but our results can be easily carried to non-binary alphabets as long as the alphabet size is no more than $\poly(n)$.  We note that the embedding result by Chakraborty et al. can be applied to strings with non-binary alphabets (see \cite{CGK16} for details).
}   Andoni et al.~\cite{AGMP13} showed that if we require the sketch for edit distance to be linear, then its size must be $\Omega(n/\alpha)$ even when we want to distinguish whether the distance is at least $2\alpha$ or no more than $2$.

Sketching is naturally related to embedding, where we want to embed the edit space to another metric space in which it is easier to compute the distance between two objects.  Ostrovsky and Rabani~\cite{OR07} gave an embedding of the edit space to the $\ell_1$ space with an $\exp{(O(\sqrt{\log n \log\log n}))}$ distortion, which was later shown to be at least $\Omega(\log n)$ by Krauthgamer and Rabani~\cite{KR09}.  In the recent work Chakraborty et al.~\cite{CGK16} obtained a weak embedding from the edit space to the Hamming space with an $O(K)$ distortion.\footnote{By ``weak'' we mean that the distortion for each pair $(s,t)$ in the edit space is bounded by $O(K)$ with constant probability.} 
\smallskip

\noindent{\em Streaming.}\ \   To the best of our knowledge, computing exact edit distance has not been studied in the streaming model.  Chakraborty et al.~\cite{CGK16} studied this problem in a variant of the streaming model where we can scan the two strings $x$ and $y$ simultaneously in the coordinated fashion, and showed that $O(K^6 \log n)$ bits of space is enough to computing $ed(x,y)$. A number of problems related to edit distance (e.g., longest increasing/common subsequence, edit distance to monotonicity) have been studied in the streaming literature~\cite{GJKK07,SW07,EJ08,GG10,CLLPTZ11,MS13}, most of which consider approximate solutions.
\smallskip

\noindent{\em Computation in RAM.}\ \    Computing edit distance in the small distance regime has been studied in the RAM model, and algorithms with $O(n + K^2)$ time have been proposed \cite{LMS98,PP09}. 
On the other hand, it has recently been shown that the general edit distance problem cannot be solved in time better than $n^{2 - \eps}$ for any constant $\eps > 0$ unless the strong Exponential Time Hypothesis is false~\cite{BI15}.

\vspace{-2mm}

\paragraph{Techniques Overview.}
We now summarize the high level ideas of our algorithms.  For simplicity the parameters used in this overview are just for illustration purposes.
\smallskip

\noindent{\em Document Exchange.}\ \  As mentioned, the document exchange problem can be solved by the algorithm of Irmak et al.~\cite{IMS05} (call it the IMS algorithm) using $O(K \log^2 n)$ bits of communication.  Intuitively speaking, IMS first converts strings $s$ and $t$ to two binary parse trees $\T_s$ and $\T_t$ respectively, where the root of the tree corresponds to the hash signature of the whole string, the left child of the root corresponds to the hash signature of the first half of the string, and the right child corresponds to that of the second half, and so on. 
IMS then tries to {\em synchronize} the two parse trees and for the receiver to identify on $\T_s$ at most $K$ root-leaf paths which lead to the at most $K$ edits. The synchronization is done using error-correcting codes at each level of the tree.  

The main idea of our new algorithm is that if we can identify those large common blocks in some optimal alignment between $s$ and $t$, then we can skip those common blocks and effectively reduce $s, t$ to two strings $s', t'$ of much smaller sizes, say, each consisting of at most $K$ substrings each of size at most $K^{99}$.  Now if we apply the IMS algorithm on $s'$ and $t'$ we only need $O(K \log^2 K^{99}) = O(K \log^2 K)$ bits of communication.  The question now is how to identify those large common blocks, which turns out to be quite non-trivial.  Note that Alice has to do this step independently in the one-way communication model, and she does not even have a good alignment between $s$ and $t$.  

Our main tool is the embedding result by Chakraborty et al.~\cite{CGK16} (denoted by the {\em CGK embedding}): we can embed binary strings $s$ and $t$ of size $n$ to binary strings $s'$ and $t'$ of size $3n$ independently, such that if $ed(s,t) = k\ (\le K)$, then with probability $0.99$ we have $\Omega(k) \le ham(x,y) \le O(k^2)$, where $ham(\cdot,\cdot)$ denotes the Hamming distance.  In the (good) case that after the embedding, the $O(k^2)$ mismatches in $s'$ and $t'$ (in the Hamming space) are distributed into at most $K$ pairs of trunks each of length at most $K^{99}$,  then we can identify those mismatched trunks as follows: We partition $s'$ and $t'$ to blocks of size $100 \sqrt{n}$, and then map those blocks back to substrings in $s$ and $t$ (the edit space). We next use an error-correcting code to identify those (up to $K$) pairs of substrings of $s$ and $t$ that differ, and recurse on those mismatched pairs.  By doing this we will have effectively reduced $s$ and $t$ to at most $K$ substrings each of length $100\sqrt{n}$ after the first round.  Then after $O(\log \log n)$ recursion rounds we can reduce the length of each of the (at most) $K$ substrings to $K^{99}$, at which moment we apply the IMS algorithm.  

The subtlety comes from the fact that if the strings $s$ and $t$ contain long common periodic substrings of sufficiently short periods, then the  $O(k^2)$ mismatches will possibly be distributed into  $O(k^2)$ remote locations in $s'$ and $t'$ (note that we cannot recurse on $k^2$ substrings since that will introduce a factor of $k^2$ in the message size). More precisely, the random walk used in the CGK embedding may get ``stuck'' in the common periodic  substrings in $s'$ and $t'$, and consequently spread the mismatches to remote locations.  We thus need to first carefully remove those long common periodic substrings in $s$ and  $t$ (again Alice has to do this independently), and then apply the above scheme to reduce the problem size.
\medskip 

\noindent{\em Sketching.}\ \  
We can view an alignment $\A$ between $s$ and $t$ as a bipartite matching, where nodes are characters in $s$ and $t$, and edges correspond to those aligned pairs $(i,j)$ in $\A$.  
The matching can naturally be viewed as a group of clusters each consisting of a set of consecutive edges $\{(i, j), (i+1, j+1), \ldots \}$, plus some singleton nodes in between.  Now let $sk(\A)$ be a sketch of $\A$ containing the first and last edges of each cluster in $\A$ plus all the singleton nodes. Intuitively, if $ed(s,t) \le K$ then for a good alignment $\A$, the size of $sk(\A)$ can be much smaller than that of $\A$.

Given a collection of matchings $\{\A_1, \ldots, \A_\rho\}$, letting $\I = \bigcap_{j \in [\rho]} \A_j$ be the set of common edges of $\A_1, \ldots, \A_\rho$, our main idea is the following: if there exists an optimal alignment that goes through all edges in $\I$, then we can produce an optimal alignment for strings $s$ and $t$ using $\{sk(\A_1), \ldots, sk(\A_\rho)\}$. 


We now again make use of the CGK embedding, which can be thought of as a random walk running on two strings $s$ and $t$ of size $n$ in the edit space, and producing two strings $s'$ and $t'$ of size $3n$ in the Hamming space such that $ham(s', t') = \poly(K \log n)$ with high probability. Each random walk consists of a set of 
states $\{(p_1, q_1), \ldots, (p_m, q_m)\}\ (p_j, q_j \in [n])$, which naturally corresponds to an alignment between $s$ and $t$.  We say a
random walk passes a pair $(u,v)$ if there exists some $j \in [m]$ such that $(p_j, q_j) = (u,v)$. Our key observation is that given $\rho = K^2 \log n$ random walks according to the CGK embedding, for any pair $(u,v)$ with $s[u] = t[v]$, if all the $\rho$ random walks pass $(u,v)$, then $(u,v)$ must be part of a particular optimal alignment (see  Lemma~\ref{lem:anchor} and its proof idea in Section~\ref{sec:anchor}).  We thus only need to compute the sketches of the alignments corresponding to those random walks, each of which corresponds to the differences between $s'$ and $t'$ in the Hamming space, for which efficient sketching algorithms have already been obtained. Moreover, the size of each sketch can be bounded by $\poly(K \log n)$ using the fact that $ham(s', t') = \poly(K \log n)$. The last step is to map these differences in the Hamming space back to the edit space for computing an optimal alignment, which requires us to add to the sketches some position-aware structures to assist the reverse mapping.  The whole sketch consists of $\rho = K^2 \log n$ sub-sketches each of size $\poly(K \log n)$, and is thus of size $\poly(K \log n)$.

\medskip 

\noindent{\em Streaming.}\ \   Our algorithm for the standard streaming model follows directly from our sketching algorithm, since the encoding phase of our sketching algorithm can be performed using a one-pass scan on the input string.  Our result in the simultaneous streaming model is obtained by implementing the classic dynamic programming algorithm for edit distance in a space and time efficient way, more precisely, by only trying to compute those must-know cells in the alignment matrix.

\vspace{-2mm}

\paragraph{Notations.}
We use $n$ as the input size, and $K$ as the threshold under which we can compute $ed(s,t)$ and the edit operations successfully.

Denote $[i..j] = \{i, i+1, \ldots, j\}$ and $[n] = [1..n]$.  When we write $x[i..j]$ for a string $x$ we mean the substring $(x[i], \ldots, x[j])$.  We use ``$\circ$'' for string concatenation.   All logs are base-$2$ unless noted otherwise. 

%% file: preliminaries.tex
\section{Preliminaries}
\label{sec:preliminary}

In this section we present some tools and previous results that we need in our algorithms.

\vspace{-2mm}
\paragraph{The CGK Embedding.}
A basic procedure in our algorithms is the embedding from edit space to Hamming space introduced in~\cite{CGK16}. The embedding is parameterized with a random string $r \in \{0,1\}^{6n}$, and maps $s \in \{0,1\}^n$ to $s' \in \{0,1\}^{3n}$. We use two counters $i$ and $j$ both initialized to $1$. Counter $i$ points to $s$ and counter $j$ points to $s'$. The algorithm proceeds in steps $j = 1, 2, \ldots$.  At step $j$ it does:
\begin{enumerate}
\vspace{-1mm}
\item $s'[j] \leftarrow s[i]$.

\vspace{-1.5mm}
\item If $r[(2j-1) + s[i]]=1$,  then $i \leftarrow i+1$.  Stop when $i = n+1$.

\vspace{-1.5mm}
\item $j \leftarrow j+1$. 
\end{enumerate}
\vspace{-1.5mm}
At the end, if $j < 3n$, then we append $(3n-j)$ `0's to $s$ to make it a string of length $3n$.  In words, at each step $j$ the algorithm reads a bit from $s$ indexed by $i$ and copies it to the output string $s'$ at position $j$.  We then decide whether to increment the index $i$ using the $((2j-1) + s[i])$-th bit of string $r$.

We are interested in comparing the Hamming distance between strings $s'$ and $t'$ produced by the embeddings on $s  \in \{0,1\}^n$ and $t  \in \{0,1\}^n$ respectively.  Let $i_0$ be a counter for $s$, $i_1$ be a counter for $t$, and $j$ be a counter denoting the current time step. At time step $j$, the bit $s[i_0]$ is copied to $s'[j]$, and the bit $t[i_1]$ is copied to $t'[j]$.  Then one of the following four cases will happen: (1) neither $i_0$ nor $i_1$ increments; (2) only $i_0$ increments; (3) only $i_1$ increments; (4) both $i_0$ and $i_1$ increment. Note that if $s[i_0]=t[i_1]$, then only first and last cases can happen, and the bits copied into $s'$ and $t'$ are the same. Otherwise if $s[i_0]\neq t[i_1]$, then the bits copied into $s'$ and $t'$  are different. Each of the four cases happens with probability $1/4$ depending on the random string $r$.  We have the following definitions.
\begin{definition}[{\bf State and Progress Step}]
We call the sequence of $(i_0, i_1)$ the {\em states} of the random walk.  We say that we have a \emph{progress step} if $s[i_0]\neq t[i_1]$ and at least one of $i_0$ and $i_1$ increments. 
\end{definition}

We can model the evolution of the ``shift'' $d = (i_0 - i_1)$ as (a different) random walk on the integer line, where at each time step, if $s[i_0]\neq t[i_1]$ then $d$ stays the same with probability $1/2$, decrements by $1$ with probability $1/4$, and increments by $1$ with probability $1/4$, otherwise if $s[i_0] = t[i_1]$ then $d$ always stays the same.
We can focus on the cases when one of $i_0$ and $i_1$ increments, and view the change of $d$ as a {\em simple} random walk on the integer line, where at each step the shift $d$ increments or decrements with equal probability.  

The following two lemmas have been shown in~\cite{CGK16} by using the properties of simple random walks.
\begin{lemma}
\label{lem:CGK}
If $ed(s,t) = k$, then for any $\ell \in \mathbb{N}^+$ the total number of progress steps in a walk according to the CGK embedding is bounded by $\ell$ with probability at least $1 - O(k/\sqrt{\ell})$. 
\end{lemma}

\begin{lemma}
\label{lem:converge}
Let $(p_0, q_0)$ be a state of a random walk $\W$ according to the CGK embedding in which $d = ed(s[p_0..n], t[q_0..n]) \ge 1$, then with probability $1 - O(1/\sqrt{\ell})$, $\W$ reaches within $\ell$ progress steps a state $(p_1, q_1)$ with $ed(s[p_1..n], t[q_1..n]) \le d - 1$.
\end{lemma}
We will also make use of the following properties of the simple random walk on the integer line where at each step the walk goes to left or right with equal probability.  The two lemmas are typically presented in the setting called the {\em gambler's ruin}.

\begin{lemma}
\label{lem:shift}
Suppose that the simple random walk starts at position $0$ and runs for $\ell$ steps, with probability $0.9$ it will not reach a position with absolute value larger than $c \sqrt{\ell}$ for a large enough constant $c$.
\end{lemma}

\begin{lemma}
\label{lem:gambler}
Suppose that the simple random walk starts at position $0$, the probability that it reaches position $b > 0$ before reaching position $-a < 0$ is at least $a/(a+b)$.
\end{lemma}

%

\vspace{-4mm}
\paragraph{The IMS Algorithm.}
As mentioned, we will use the IMS algorithm proposed in \cite{IMS05}.  In fact, we can slightly improve the original IMS algorithm in \cite{IMS05} to get the following result.  
\begin{theorem}
\label{thm:IMS}
There exists an algorithm for document exchange having communication cost $O(K(\log K+\log \log n)\log n))$, running time $\tilde{O}(n)$, and success probability $1 - 1/\poly(K \log n)$, where $n$ is the input size and $K$ is the distance upper bound.
\end{theorem}
We delay the description of our improved IMS and the analysis to Appendix~\ref{proof:thm:IMS}.

%

\vspace{-4mm}
\paragraph{Periods and Random Walk.}
When performing the CGK embedding, common periods in the strings may ``slow down'' the random walk, and consequently distribute the mismatches into remote locations (which is undesirable for our algorithm for document exchange; see the techniques overview in the introduction). On the other hand,  if two strings do not share long periodic substrings of small periods, then the random walk induced by the embedding will make steady progress.  We have the following lemma whose proof is delayed to Appendix~\ref{proof:lem:walk-break}.

\begin{lemma}
\label{lem:walk-break}
Suppose that we have $w=s[i..i+m]=t[j..j+m]$ and that the substring $w$ has no substring of length $\ell<m$ with period at most $\theta$. Then, the random walk induced by the CGK embeddings of $s$ and $t$ will have the following property: suppose at a given step the random walk is in the state $(p,q)$ such that $|(p-i)-(q-j)|\in[1..\theta]$, $p\in [i..i+m-\ell]$ and $q\in [j..j+m-\ell]$, then there will be a progress step for some pair $(u_0,u_1)$ with $u_0\in [p,p+\ell-1]$ and $u_1=u_0+(q-p)$. 
\end{lemma}


\vspace{-4mm}
\paragraph{Error-Correcting.}
We consider the following problem: Alice holds a vector $a$ and Bob holds a vector $b$, where $a$ and $b$ are of length $u$ over an alphabet of size $\sigma$. Let $k$ be a given threshold.   Bob knows that the differences between $a$ and $b$ fall into a set $S$ of $\lambda \ge k$ coordinates, but Alice does not know the set $S$. The task is for Alice to send a message $msg$ to Bob so that Bob can recover $a$ exactly based on $msg$ and his input string $b$ if $ham(a,b) \le k$, and output ``error'' otherwise.  We have the following lemma whose proof is delayed to Appendix~\ref{proof:lem:ECC}.

\begin{lemma}
\label{lem:ECC}

There exists an algorithm for the above problem having communication cost $O(\log\log u + k(\log\sigma + \log \lambda + \log (1/p))$, running time $\tilde{O}(u +  k\lambda)$, and success probability $1 - p$.
\end{lemma}

\vspace{-4mm}
\paragraph{Sketching Hamming Distance.}
We will use the result for sketching Hamming distance as a building block in our sketching algorithm.  In particular we will use the one in~\cite{porat2007improved}, and state it for a general alphabet. 

\begin{lemma}[\cite{porat2007improved}]
\label{lem:ham}
There exists a sketching algorithm which, given a vector of length $n$ over an alphabet of size $\sigma=O(\poly(n))$ and a threshold $k$, outputs a vector of length $O(k \log n)$, such that given the sketches of two vectors $a$ and $b$, one can recover with probability $(1 - 1/\poly(n))$ the coordinates (indices and values) where they differ in time $O(k \log n)$ if $ham(a,b) \le k$, and outputs ``error" otherwise. The sketching process can be done in the one-pass streaming fashion and runs in $O(\log n)$ time per element.
\end{lemma}

%

\vspace{-4mm}
\paragraph{Reducing Random Bits.}
In our algorithms Alice and Bob may use many shared random bits for the CGK embeddings. We can use the standard probabilistic method to reduce the total number of random bits to $O(\log n)$ in the non-uniform case. To reduce the number of random bits in the uniform case, we can also make use of Nisan's pseudo-random generator~\cite{nisan1992pseudorandom}, which states that any bounded-space randomized algorithm for a decision problem can be solved using $O(s\log t)$ random bits where $s$ is the space used by the algorithm and $t$ is the total running time of the algorithm. The random bit generation in Nisan's pseudo-random generator can be done in the streaming fashion and in linear time. We refer readers to \cite{CGK16} for more details on reducing random bits.  

%% file: doc-exchange.tex
\section{Document Exchange}
\label{sec:doc}

In this section we prove the following theorem.
\begin{theorem}
\label{thm:doc}

There exists an algorithm for document exchange having communication cost $O(K (\log^2 K + \log n))$, running time $\tilde{O}(n)$, and success probability $1 - 1/\poly(K \log n)$, where $n$ is the input size and $K$ is the distance upper bound.
\end{theorem}


\subsection{The Algorithm}
\label{sec:alg}

Let $k = ed(s,t)$ be the edit distance between $s$ and $t$ that we want to compute.  Our algorithm for document exchange consists of two stages.  
After the first stage we effectively reduce the problem to a much smaller size, more precisely, to that on two strings each consisting of at most $k\leq K$ substrings each of size at most $(K 2^{\sqrt{\log n}})^{O(1)}$, while preserving the edit distance.
We will then in the second stage run the IMS algorithm on the reduced problem to compute the distance and all the edits. We thus only describe the first stage of the algorithm.

The first stage consists of $L = O(\log \log n)$ levels, each of which consists of two phases.  These levels and phases are executed in synchrony between Alice and Bob, but the whole communication is still one-way.   We now describe the two phases at each level $\ell \in [L]$. We will use the following parameters, which are known to both parties before running the algorithm.  Let $c_1, c_2$ be two sufficiently large constants.
\begin{itemize}
\vspace{-1mm}
\item $b_\ell = \sqrt{n_\ell} 2^{c_1 (\log K + \sqrt{\log n})}$: block size in the first phase.

\vspace{-1mm}
\item $b'_\ell = \sqrt{n_\ell} 2^{(c_1+c_2) (\log K + \sqrt{\log n})}$: block size in the second phase.

\vspace{-1mm}
\item $\theta_\ell = b_\ell/2$: upper bound of period length.
\end{itemize}
Let $n_\ell$ be an upper bound of the effective size of the problem at the beginning of the $\ell$-th level. $n_1 = n = \abs{s} = \abs{t}$, and $n_\ell = K b'_{\ell -1}\ (\ell \ge 2)$.

\vspace{-2mm}

\paragraph{Phase I:}  
Alice partitions her string $x$ to blocks of size $b_\ell$, except that the first block is of a uniformly random size $\Delta_\ell \in [b_\ell/2, b_\ell]$, and the last block is of size in the range $[1..b_\ell]$.
Denote these blocks by $B_1, \ldots, B_{m}$. Alice sends $\Delta_\ell$ to Bob.

Next, Alice creates a vector $U_\ell = (e_1, \ldots, e_{m})$ where $e_j = (\chi_j, w_j)$ contains the following information of the $j$-th block $B_j$:  If $B_j$ is part of a periodic substring with period length at most $\theta_\ell$, then we set $\chi_j = 1$ and $w_j$ to be the period length;\footnote{We say that substring $s[i..j]$ is part of periodic substring of $s$ with period $w_j>0$ iff $s[i..j]=s[i-w_j..j-w_j]$ and $w_j$ is the smallest number with this property.} otherwise we set $\chi_j = 0$ and $w_j = 0$.   Alice then sends to Bob a redundancy of $U_\ell$ (denoted as $sk(U_\ell)$) that allows to recover up to $2K$ errors using the scheme in Lemma~\ref{lem:ECC} (setting $p = 1/\poly(K \log n)$).

Bob maintains a vector $\tilde{x}$ based on his decoding of Alice's string $x$ in the previous $(\ell-1)$ levels such that with high probability $\tilde{x}$ only differs from $x$ at no more than $n_\ell$ coordinates (those that Bob hasn't recovered), for which Bob marks `$\perp$'. ($\tilde{x}$ is initialized to be an all-`$\perp$' vector of length $n$ at the beginning of the algorithm.)  After receiving the offset $\Delta_\ell$ he does the same partitioning on both $\tilde{x}$ and $y$, getting $(P_1, \ldots, P_m)$ and $(Q_1, \ldots, Q_m)$ respectively.  He then examines each block $P_j$. If $P_j$ or the $\theta_\ell$ positions in $\tilde{x}$ preceding $P_j$ contain a `$\perp$', then he builds $e'_j = (\chi'_j, w'_j)$ the same way as $e_j = (\chi_j, w_j)$ by looking at $Q_j$'s context in $y$; otherwise he builds $e'_j = (\chi'_j, w'_j)$ the same way as $e_j = (\chi_j, w_j)$ by looking at $P_j$'s context in $\tilde{x}$.  Let $U'_\ell = (e'_1, \ldots, e'_m)$.  We will show in the analysis that Bob can recover Alice's vector $U_\ell$ using $sk(U_\ell)$ and his vector $U'_\ell$ with high probability. 

\begin{figure}[t]
\centering
\includegraphics[height = 0.6in]{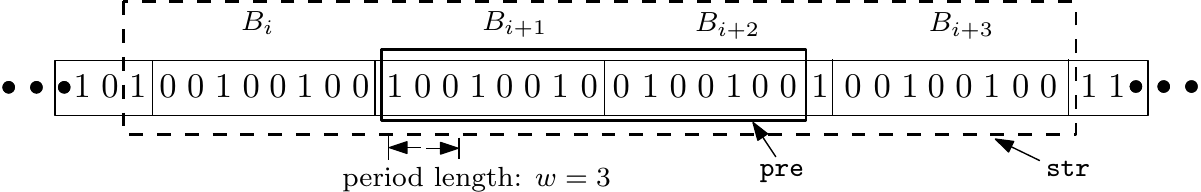}
\caption{Periods are removed in Phase I.}
\label{fig:remove-period}
\end{figure}

Finally, for the consecutive blocks contained in each periodic substring $\str$ of $x$, denoted by $B_1 \circ  B_2 \cdots \circ B_{z-1} \circ B_z$, Alice removes the longest prefix $\pre$ of $B_2 \circ \cdots \circ B_{z-1}$ (i.e., excluding the first block $B_1$ and the last block $B_z$) such that $\abs{\pre}$ is a multiple of $w$ where $w$ is the length of the period of $\str$.\footnote{Formally two consecutive blocks $B_{j}$ and $B_{j+1}$ are contained in the same periodic substring iff $w_j=w_{j+1}$.}   See Figure~\ref{fig:remove-period} for an illustration.  Bob, after decoding Alice's vector $U_\ell$, does the same thing, that is, he removes in his input string $y$ and the maintained string $\tilde{x}$  those characters of the same indices that Alice removes.

\vspace{-3mm}
\paragraph{Phase II:}  
Alice maps her string $x$ into the Hamming space using the CGK embedding (see Section~\ref{sec:preliminary}), getting a string $x'$, and then partitions $x'$ to blocks of size $b'_\ell$, except that the first block is of a uniformly random size $\Delta'_\ell \in [b'_\ell/2, b'_\ell]$, and the last block is of size in the range $[1..b'_\ell]$. 
She then maps these blocks back to the edit space, getting a partition of the original string $x$. Denote these blocks by $A_1, \ldots, A_{d}$.  Alice sends $\Delta'_\ell$ to Bob, and Bob does the same partitioning to his string $y$, getting $A'_1, \ldots, A'_d$.

Next, Alice creates a vector $V_\ell = (g_1, \ldots, g_{d})$ where $g_j = (h(A_j), r_j, \calE_j)$ contains the following information on the $j$-th block $A_j$: $h : \{0,1\}^* \to \left[(K n_i)^{\Theta(1)}\right]$ is a Karp-Rabin hash signature of $A_j$; $r_j$ is the length of $A_j$; and $\calE_j = 1$ if the first character of $A_j$ is shared with the last character of $A_{j-1}$ (this could happen due to the copy operations in the CGK embedding; $\calE_j$ is needed for Bob to identify the boundaries of the unmatched substrings in $x$ accurately), and $\calE_j = 0$ otherwise.  Alice then sends to Bob a redundancy of $V_\ell$ (denoted as $sk(V_\ell)$) that allows to recover up to $K$ errors using the scheme in Lemma~\ref{lem:ECC} (setting $p = 1/\poly(K \log n)$).  Bob does the same for $(A'_1, \ldots, A'_d)$, getting a vector $V'_\ell = (g'_1, \ldots, g'_d)$. We will show in the analysis that Bob can recover Alice's vector $V_\ell$ using $sk(V_\ell)$ and his vector $V'_\ell$ with high probability.  Bob then updates the string $\tilde{x}$ by replacing the `$\perp$'s with actual contents for those matched blocks, and
at the next level he will do the decoding recursively on those unmatched blocks (i.e., those still marked with `$\perp$') whose sizes sum up to no more than $n_{\ell+1}= K b'_\ell$. 


The first stage concludes when the length of blocks in the second phase becomes $b'_\ell \le (K 2^{\sqrt{\log n}})^{10(c_1+c_2)}$, from where Alice and Bob apply IMS directly to compute the edit distance and all the edits.  The message Alice sends to Bob in the first stage includes the offsets $\{\Delta_\ell, \Delta'_\ell\ |\ \ell \in [L] \}$ and the redundancies $\{sk(U_\ell), sk(V_\ell)\ |\ \ell \in [L]\}$.  Note that Alice can compute these independently using parameters $b_\ell, b'_\ell$ and $\theta_\ell$ at each level $\ell \in [L]$.  In the second stage, Alice sends Bob the IMS sketch on her string $x$, but omits all the top-levels in the IMS sketch  at which the block sizes are larger than $b'_L$, since Bob does {\em not} need to do the recovery at those levels given the first stage. 
At the end, Bob can recover Alice's input string $s$ by adding back the removed periods at all levels. 

\subsection{The Analysis}
\label{sec:analysis}

\paragraph{Correctness.}  
We focus on the case $k = ed(s,t) \le K$; otherwise if $k > K$ then Bob can detect this during the decoding (in particular, various recoveries using Lemma~\ref{lem:ECC} will report ``error'') and output ``error'' with probability $1 - 1/\poly(K \log n)$.  

We establish few lemmas.  The first lemma shows that the period-removal step in Phase I at each level preserves the edit distance.  The proof is technical and is delayed to Appendix~\ref{proof:lem:period_elim}.
\begin{lemma}
\label{lem:period_elim}
Given two strings $s=ppp$ and $t$ of the same length, letting $\pi=|p|\leq ed(s,t)$ be the length of the period of $s$, the edit distance between $s'=pp$ and $t'=t[1..\pi] \circ t[2\pi+1..3\pi]$ is at most $ed(s,t)$.
\end{lemma}

The following two lemmas show that the redundancies sent by Alice in the two phases are sufficient for Bob to recover Alice's vectors $U_\ell$ and $V_\ell$. 
\begin{lemma}
\label{lem:remove-period}
At each level $\ell$ in Phase I, let $x$ and $y$ be the strings held by Alice and Bob respectively. Suppose that $ed(x,y) \le K$, then with probability $1 - 1/\poly(K \log n)$,  Bob can recover Alice's vector $U_\ell$ using his vector $U'_\ell$ and Alice's message $sk(U_\ell)$.   Moreover, after the period-removal step the edit distance between the two resulting strings $x'$ and $y'$ does not increase, that is, $ed(x',y') \le ed(x,y)$.
\end{lemma}

\begin{proof}
For the first part of the lemma, we just need to show that $U_\ell$ and $U'_\ell$ differ in at most $2K$ pairs $(e_j, e'_j)$.  Note that we only need to look at those $e'_j$ built from $Q_j$ in $y$, since otherwise if  $e'_j$ is built from $P_j$ in $\tilde{x}$ then we always have $e_j = e'_j$.  We thus only need to consider at most $n_\ell / \theta_\ell (= 2n_\ell / b_\ell)$ pairs $(e_j, e'_j)$.

We call a pair of block $(B_j, Q_j)$ {\em good} if there is no edit in both $B_j$ and $Q_j$ as well as their preceding $\theta_\ell$ characters in $x$ and $y$ respectively; we call the pair {\em bad} otherwise.  Since $ed(x,y) \le K$, there are at most $2K$ bad pairs.  
We call a pair $j$ {\em periodic} if $\chi_j = 1$ or $\chi'_j = 1$ (i.e., at least one of $B_j$ or $Q_j$ is part of a periodic subtring in $s$ or $t$ with period length at most $\theta_\ell$), and {\em non-periodic} otherwise.  Clearly for a good and non-periodic pair $j$ we must have $(\chi_j, w_j) = (\chi'_j, w'_j) = (0,0)$.

We now show that for a good pair $j$, if $\chi_j = 1$ or $\chi'_j = 1$, then with probability $1 - O(K/b_\ell)$ we have $e_j = e'_j$.  We prove only for the case when $\chi_j = 1$; the proof for the other case is the same. The observation is that we have a random shift $\Delta_\ell \in [b_\ell/2, b_\ell]$ in the block partition, and in any optimal alignment the indices of two matching characters in $x$ and $y$ can differ by at most $ed(x,y) \le K$, thus for a good block $B_j$ that is part of a periodic substring of period length at most $\theta_\ell = b_\ell/2$ in $x$, with probability $1 - O(K/b_\ell)$, $Q_j$ is also part of a substring with the same period.  By a union bound on at most $2n_\ell / b_\ell$ pairs $(e_j, e'_j)$, we have that with probability 
$1 - 2n_\ell/b_\ell \cdot K/b_\ell \ge 1 - 1/\poly(K\log n)$,
for all good and periodic pairs $j$, $e_j = e'_j$.  The first part of the lemma follows.

The second part of the lemma follows directly from Lemma~\ref{lem:period_elim}. Note that when removing periods in each periodic substring $\str$ we have kept the first and the last blocks that are contained in $\str$, and thus Bob can recover those periods that have been removed. 
\end{proof}

\begin{lemma}
\label{lem:remove-common}

At each level $\ell$ in Phase II, let $x$ and $y$ be the strings held by Alice and Bob respectively. Suppose that $ed(x,y) \le K$, then with probability $1 - 1/\poly(K \log n)$,  Bob can recover Alice's vector $V_\ell$ using his vector $V'_\ell$ and Alice's message $sk(V_\ell)$.
\end{lemma}


\begin{proof}
We again just need to show that $V_\ell$ and $V'_\ell$ differ on at most $K$ pairs $(g_j, g'_j)$. 

Let $x'$ and $y'$  be two strings after performing the CGK embedding on $x$ and $y$ respectively.  Let $\W$ be the random walk corresponding to the embedding.  Recall that starting from a state $(p_0, q_0)$ where a progress step happens,  by Lemma~\ref{lem:converge} we have that with probability $1 - O(1/\sqrt{\gamma})$,  after at most $\gamma$ progress steps $\W$ will reach a state $(p_1, q_1)$ such that $ed(x[p_1..n], y[q_1..n]) \le ed(x[p_0..n], y[q_0..n]) - 1$.  We call such a sequence of walk steps a {\em progress phase}.  Since $ed(x,y) \le K$, the total number of progress phases is upper bounded by $K$.   By a union bound, with probability $1 - O(K/\sqrt{\gamma})$, each of the  at most $K$ progress phases ``consumes'' at most $\gamma$ progress steps.   Note that for all indices $j$ between two progress phases, we have $x'[j] = y'[j]$, that is, the two substrings of $x'$ and $y'$ are perfectly matched.

The key observation is that after the period-removal process in Phase I, by Lemma~\ref{lem:walk-break} we will have a ``break'' after passing at most two blocks (less than $4b_\ell$ characters when counting the periods crossing the two block boundaries) allowing for at least one progress step to execute.  Therefore the total number of pairs of coordinates in $x'$ and $y'$ that are involved in one of the at most $K$ progress phases can be bounded by $K \cdot \gamma \cdot 4b_\ell$ with probability $1 - O(K/\sqrt{\gamma})$.  Setting $\gamma = (K \log n)^{c_2/2}$.  By our choices of parameters we have
$$b'_\ell \ge  2^{c_2 (\log K + \sqrt{\log n})} b_\ell \ge (K \cdot \gamma \cdot 4b_\ell) \cdot (K\log n)^{\Theta(1)}.$$ 
Also recall that there is a random shift $\Delta'_\ell \in [b'_\ell/2, b'_\ell]$ at the beginning of the block partition in Phase II.  We thus have with probability at least 
$$1 - \left( O(K/{(K \log n)^{c_2/4}}) + (K \cdot \gamma  \cdot 4b_\ell)/(b'_\ell/2) \right) \ge 1 - {1}/{\poly(K\log n)}
$$ 
that at most $K$ pairs $(A_j, A'_j)$ differ, where $(A_1, \ldots, A_d)$ is the block partition of $x$ in Phase II, and $(A'_1, \ldots, A'_d)$ is that of $y$.  The lemma follows.
\end{proof}

The correctness of the algorithm follows from Lemma~\ref{lem:remove-period} and Lemma~\ref{lem:remove-common}.  Note that the first stage will finish in at most $O(\log \log n)$ levels since $\log n_{\ell}$ decreases at each level by a factor of 
\begin{eqnarray*}
{\log n_\ell}/{\log n_{\ell+1}} = {\log n_\ell}/{\log(K b'_\ell)} &\ge& {\log n_\ell}\left/{\left(\log \left(\sqrt{n_\ell}2^{(c_1 + c_2) (\log K + \sqrt{\log n})} \right) \right)} \right.  \ge 1.5,
\end{eqnarray*}
where in the last inequality we used the fact that $n_\ell \ge (K 2^{\sqrt{\log n}})^{10(c_1+c_2)}$ (otherwise we go to the second stage and apply the IMS algorithm).  The overall success probability is at least
$$1 - \left(1/\poly(K \log n) + 1/\poly(K \log n)\right) \cdot O(\log\log n) - 1/\poly(K \log n) \ge 1 - 1/\poly(K \log n),
$$
where the last term on the left hand side counts the error probability of the IMS algorithm (Theorem~\ref{thm:IMS}).

%
%

\vspace{-4mm}
\paragraph{Complexities.}
We now analyze the communication cost and the running time of our algorithm.  In the first stage, at each level $\ell$, the communication cost is dominated by the size of $sk(U_\ell)$ and $sk(V_\ell)$, both of which can be bounded by $O(K \log n_\ell)$ by Lemma~\ref{lem:ECC} (where $\lambda \le n_\ell$). Since $\log n_\ell$ decreases by a constant factor at each level, the total size of the message in the first stage is bounded by $O(K \log n)$.  In the second stage,  the total number of levels in the IMS sketch is bounded by $\log b'_L = O(\log K + \sqrt{\log n})$, and the sketch size per level is $O(K(\log K + \log \log n))$ (these are the same as Theorem~\ref{thm:IMS} except that the number of levels has been reduced from $\log n$ to $\log b'_L$). Thus the total size of the IMS sketch is bounded by $O(\log K + \sqrt{\log n}) \cdot O(K(\log K + \log \log n)) = O(K(\log^2 K + \log n))$.  Summing up, the total communication is bounded by $O(K(\log^2 K + \log n))$.  

The running time of the first stage is again dominated by the encoding and decoding time of the scheme in Lemma~\ref{lem:ECC}, which is bounded by $\tilde{O}(n)$. The running time of the IMS algorithm in the second stage is also bounded by $\tilde{O}(n)$ (Theorem~\ref{thm:IMS}).

%% file: sketching.tex
\section{Sketching}
\label{sec:sketch}

In this section we show the following theorem.  
\begin{theorem}
\label{thm:sketching}

There exists a sketching algorithm for computing edit distance and all the edit operations having sketch size $O(K^8 \log^5 n)$, encoding time $\tilde{O}(K^2 n)$, decoding time $\poly(K \log n)$, and success probability $0.9$, where  $n$ is the input size and $K$ is the distance upper bound. 
When the distance is above the upper bound $K$, the decoding algorithm outputs ``error'' with probability $1-1/\poly(n)$. 
\end{theorem}
Note that the success probability can be boosted to high probability by using the standard parallel repetition (and then take the one with the smallest distance).

\subsection{The Algorithm and Analysis}
We first introduce a concept called {\em effective alignment}, and then show that a set of effective alignments satisfying a certain property can be used to construct an optimal alignment between the two strings.  

\begin{definition}[{\bf Effective Alignment}]
\label{def:effective}

Given two strings $s, t \in \{0,1\}^n$, we define an effective alignment between $s$ and $t$ as a triplet $(G, g_s, g_t)$, where
\vspace{-1mm}
\begin{itemize}
\item $G  = (V_s, V_t, E)$ is a bipartite graph where nodes $V_s = \{1, \ldots, n\}$ and $V_t = \{1, \ldots, n\}$ correspond to indices of characters in $s$ and $t$ respectively, and if $(i, j) \in E$ then $s[i] = t[j]$.  Moreover,  edges in $E$ are non-crossing, that is, for every pair of edges $(i,j)$ and $(i', j')$, we have $i < i'$ iff $j < j'$.

\vspace{-1mm}
\item $g_s$ is a partial function defined on the set of singletons (unmatched nodes) $U_s \subseteq V_s$; for each $i \in U_s$, define $g_s(i) = s[i]$. Similarly, $g_t$ is a partial function defined on the set of singletons $U_t \subseteq V_t$; for each $j \in U_t$, define  $g_t(i) = t[i]$.
\end{itemize} 
\end{definition}
Intuitively, an effective alignment can be seen as a {\em summary} of an alignment after removing the information of those matched nodes.  The following lemma gives the main idea of our sketching algorithm.
\begin{lemma}
\label{lem:sketch-main}

We can compute an optimal alignment for $s$ and $t$ using a set of effective alignments under the promise that there exists an optimal alignment going through all edges that are common to all  effective alignments.
\end{lemma}

\begin{proof}
Let $\I$ be the set of edges that are common to all effective alignments. We will show that we can reconstruct using these effective alignments all characters $s[i]$'s and $t[j]$'s for which $i, j$ that are not adjacent to any edges in $\I$.  

For convenience we add two dummy edges $(0,0), (n+1, n+1)$ to all the effective alignments to form the boundaries.   We can view the edges in $\I$ as a group of clusters $\C_1, \ldots, \C_\kappa$ (counting from left to right) each of which consists of a set of consecutive matching edges $\{(i, j), (i+1, j+1), \ldots \}$, plus some singleton nodes  between these clusters.   Now consider a particular effective alignment $\A$ and an $\ell \in [\kappa - 1]$.  Let $(a_1, b_1), \ldots, (a_z, b_z)$ be the set of edges in $\A$ that lie between $\C_\ell$ and $\C_{\ell+1}$.  By the definition of the effective alignment we can learn directly from $\A$ all the singletons in $\A$ that lie between $\C_\ell$ and $\C_{\ell+1}$.  It remains to show that we can recover the characters in $s$ and $t$ that correspond to the nodes $a_1, \ldots, a_z$ and $b_1, \ldots, b_z$.  For convenience we will identify nodes and their corresponding characters in the strings.

 Let us consider edges $(a_1, b_1), \ldots, (a_z, b_z)$ one by one from left to right.  Note that for each $x \in [z]$, we just need to recover one of $s[a_x]$ and $t[b_x]$ because they are equal.  Since $(a_x, b_x)$ is not in $\I$, we know that there exists another effective alignment $\A'$ which does not contain $(a_x, b_x)$.  We have the following cases:
\vspace{-1mm}
\begin{enumerate}
\item $a_x$ is a singleton in $\A'$.  In this case we can recover $s[a_x]$ directly from $\A'$.

\vspace{-1mm}
\item $a_x$ is connected to a node $u$ in $\A'$ such that $u < b_x$.  This case is again easy since $t[u]$ has already been recovered, and thus we just need to set $s[a_x] = t[u]$.

\vspace{-1mm}
\item $a_x$ is connected to a node $u$ in $\A'$ such that  $u > b_x$.  In this case $b_x$ is either a singleton or is connected to a node $v < a_x$ in $\A'$. In the former case we can recover $t[b_x]$ directly from $\A'$, and in the latter case since we have already recovered $s[v]$,  we can just set $t[b_x] = s[v]$.
\end{enumerate}
\vspace{-1mm}

We thus have shown that we can recover all nodes that are not adjacent to any edges in $\I$.  Since by the promise that there exists an optimal alignment containing the edges that are common to all effective alignments (i.e., $\I$), we can construct such an optimal alignment by aligning characters in $s$ and $t$ in the gaps between clusters $\C_1, \ldots, \C_\kappa \subseteq \I$ in the optimal way.
\end{proof}

The rest of our task is to design sketches $sk(s)$ and $sk(t)$ for $s$ and $t$ respectively so that using $sk(s)$ and $sk(t)$ we can extract a set of effective alignments satisfying the promise in Lemma~\ref{lem:sketch-main}.  Intuitively, the size of $sk(s)$ and $sk(t)$ can be small if (1) the information contained in each effective alignment is small, and (2) the number of effective alignments needed is small.  Our plan is to construct $\rho = \poly(K \log n)$ effective alignments $\A_1, \ldots, \A_\rho$, each of which only contains $\poly(K \log n)$ singletons, and there is an optimal alignment going through all edges in $\I = \bigcap_{j \in [\rho]} \A_j$.  Note that we can compress the information of consecutive edges in each $\A_j$ by just writing down the first and the last edges, whose number is bounded by the number of singletons. We thus can bound the sketch size by $\poly(K \log n)$. In the rest of this section we show how to carry out this plan.
\smallskip

We again make use of the random walk in the CGK embedding.  Recall that a random walk on the two strings $s$ and $t$ can be represented as $\W = ((p_1, q_1), \ldots, (p_m, q_m))$ where $(p_j, q_j)\ (p_j, q_j \in [n])$ are states of $\W$.  $\W$ naturally corresponds an alignment (not necessarily optimal) between $s$ and $t$. 
More precisely, $\W$ corresponds to the alignment $\A$ constructed greedily by adding states $(p_m, q_m), \ldots, (p_1, q_1)$ as edges one by one whenever $s[p_j] = t[q_j]$~\footnote{Notice that the same $p_j$ (resp. $q_j$) can not appear in two distinct edges since state $(p_j,q_j)$ is always 
followed by state $(p_j+1,q_j+1)$ whenever $s[p_j]=t[q_j]$.}. We will give a detailed description of this connection in the proof of Lemma~\ref{lem:recover}.

Set $\rho = c_\rho K^2 \log n$ for a large enough constant $c_\rho$, and $N = c^4_\rho K^6 \log^2 n$. 
We say that a random walk {\em passes} a pair $(u,v)$ if there exists some $j \in [m]$ such that $(p_j, q_j) = (u,v)$.  We call a random walk $\W$ {\em good} if the total number of progress steps in $\W$ is at most $N$ (recall that we have a progress step at state $(p_j, q_j)$ only if $s[p_j] \neq t[q_j]$).  By Lemma~\ref{lem:CGK} a random walk is good with probability $1 - O(1/(c_\rho^2 K^2 \log n))$.   We will generate $\rho$ random walks. By a union bound, we have
\begin{claim}
\label{cla:good-walk}
The probability that all the $\rho$ random walks are good is at least $0.99$.
\end{claim}

The following lemma indicates that the alignments corresponding to a set of $\rho$ random walks can be used (after compression) as a set of effective alignments satisfying the promise in Lemma~\ref{lem:sketch-main}.  The proof of the lemma is technical and is presented in Section~\ref{sec:anchor}.  

\begin{definition}[{\bf Anchor}]
Given $\rho$ random walks generated according to the CGK embedding, we say that a pair $(u,v)\ (u,v \in [n])$ is an {\em anchor} if $s[u] = t[v]$, and all the $\rho$ random walks pass $(u,v)$. 
\end{definition}

\begin{lemma}
\label{lem:anchor}
With probability $1 - 1/n^2$, there is an optimal alignment going through all anchors.
\end{lemma}


\vspace{-4mm}

\paragraph{The Sketch.}  
We now show how to design sketches for $s$ and $t$ from which we can extract the $\rho$ effective alignments corresponding to the $\rho$ random walks.  For recovering each of the $\rho$ effective alignments, we prepare a pair of structures we call the {\em hierarchical} structure and the {\em content} structure, as follows.  Let $B = 4\log n$ be a parameter denoting a basic block size.  
\vspace{-0.5mm}
\begin{enumerate}
\item {\em The hierarchical structure $P$.} \ \  
Let $s' \in \{0,1\}^{3n}$ be the image of $s \in \{0,1\}^n$ after the CGK embedding. W.l.o.g.\ assume $3n/B$ is a power of $2$ (otherwise we can pad $0$s to the image $s'$). We build a binary tree of depth $L = \log (3n/B)$ on top of $s'$, whose leaves (at level $1$) correspond to a partition of $s'$ into blocks of size $B$, and internal nodes at level $\ell$ correspond to substrings of $s'$ of size $2^\ell B$ (i.e., the concatenation of the blocks corresponding to all the leaves in its subtree).  For each level $\ell \in [L]$, setting $d_\ell = 3n/(2^\ell B)$,  we create a vector $V_\ell = ((h_1, \eta_1), \ldots, (h_{d_\ell}, \eta_{d_\ell}))$ where $h_j$ is a hash signature of the pre-image (in $s$) of the substring in $s'$ corresponding to the $j$-th node at level $\ell$, and $\eta_j$ is the length of the pre-image.  We then build a sketch $P_\ell$ of $V_\ell$ that allows to recover up to $N$ errors using the scheme in Lemma~\ref{lem:ham}.  Let $P = (P_1, \ldots, P_{L})$.  The size of $P$ is $O(L \cdot N \log n) = O(N \log^2 n)$ by Lemma~\ref{lem:ham}.

\vspace{-1mm}
\item {\em The content structure $Q$.}  \ \ 
Again let $s$ and $s'$ be repsectively the original string and the string after the embedding.  We partition $s'$ into blocks of size $B$,  and create a vector $U = (x_1, \ldots, x_{3n/B})$ where $x_j$ is the pre-image (in $s$) of the $j$-th block of $s'$.
We then build a sketch $Q$ of $U$ that allows to recover up to  $N$ errors using the scheme in Lemma~\ref{lem:ham}.  The size of $Q$ is $O(N (B + \log n)) = O(N \log n)$ by Lemma~\ref{lem:ham}.
\end{enumerate}
The final sketch $sk(s)$ for $s$ consists of $\rho$ independent copies of $(P, Q)$.  Clearly the size of $sk(s)$ is bounded by $\rho \cdot O(N \log^2 n) = O(K^8 \log^5 n)$.   Similarly, the sketch $sk(t)$ for $t$ consists of $\rho$ independent copies $(P', Q')$, each of which is constructed in the same way as $(P, Q)$ (but for string $t$).  The time for computing the sketch is bounded by $\tilde{O}(\rho \cdot n) = \tilde{O}(K^2 n)$.

Now we show that we can extract an effective alignment for $s$ and $t$ from $(P, Q)$ and $(P', Q')$.  We again focus on the case when $k = ed(s,t) \le K$.  Otherwise if $k > K$ then various recoveries using Lemma~\ref{lem:ham} in the decoding will report error with probability $1 - 1/\poly(n)$.

\begin{lemma}
\label{lem:recover}

Given $(P, Q)$ and $(P', Q')$,  with probability at least $0.98$, we can extract an effective alignment $\A$ for $s$ and $t$ corresponding to a random walk $\W$ according to the CGK embedding such that for any edge $(u,v) \in \A$ there exists some state $(p,q) \in \W$ such that $(p,q) = (u,v)$.
\end{lemma}

\begin{proof} (sketch)
First by Claim~\ref{cla:good-walk} we know that with probability $0.99$ all random walks are good, conditioned on which we can compute the differences between $(P, Q)$ and $(P', Q')$ successfully with probability $1 - o(1)$.  Note that each pair of mismatched blocks between $s$ and $t$ corresponds to two root-leaf paths in the binary trees (in the hierarchical structures) constructed for $s$ and $t$  where the contents of nodes differ.  By computing the differences between $P$ and $P'$ (using the hash signatures), we can use the length information recovered at each level of the root-leaf paths to find in $s$ and $t$ the locations of the at most $N$ blocks where they differ.  Next, by computing the differences between $Q$ and $Q'$ we can fill the actual contents of those mismatched blocks. Finally, after getting the positions and contents of mismatched blocks, we can add edges corresponding to states $(p_m, q_m), \ldots, (p_1, q_1)$ with $s[p_j] = t[q_j]$ in $\W$ in a greedy fashion to obtain an effective alignment for $s$ and $t$.  The full proof is given in Appendix~\ref{proof:lem:recover}.
\end{proof}

By Lemma~\ref{lem:anchor} and Lemma~\ref{lem:recover} we have the following immediate corollary.
\begin{corollary}
\label{cor:walk}
With probability $0.97$, we can extract $\rho$ effective alignments  from $sk(s)$ and $sk(t)$ such that there is an optimal alignment between $s$ and $t$ going through all edges that are common to all of the $\rho$ alignments.
\end{corollary}

The next lemma finishes the proof of Theorem~\ref{thm:sketching}.  We comment that the fact that the decoding time can be reduced to $\poly(K \log n)$ is very useful in the distributed/parallel computation models where the decoding phase is performed in a central server which collects sketches produced from a number of machines.
\begin{lemma}
\label{lem:decode}
With probability $0.97$, we can extract $\rho$ effective alignments  from $sk(s)$ and $sk(t)$ and use them to construct an optimal alignment between $s$ and $t$.  The running time of the construction is $\poly(K \log n)$.
\end{lemma}

\begin{proof}(sketch)
By Corollary~\ref{cor:walk}, we know that with probability $0.97$ we can extract $\rho$ alignments from sketches $sk(s)$ and $sk(t)$ satisfying the promise of Lemma~\ref{lem:sketch-main}. The construction of the optimal alignment between $s$ and $t$ basically follows from the arguments in the proof of Lemma~\ref{lem:sketch-main}.  We first compute the set of common edges $\I$ of the $\rho$ effective alignments, which can be done in $\poly(K \log n)$ time by a $\rho$-way merging.  Recall that all the edges in $\I$ will be in the optimal alignment although we may not know the values of their adjacent nodes. We then try to recover the rest of the nodes of $s$ and $t$ that are not adjacent to any edges in $\I$ (call them the {\em remaining} nodes), and match them in the optimal way.  The main challenge is that directly recovering all the remaining nodes using the argument in the proof of Lemma~\ref{lem:sketch-main} may take time more than $\poly(K \log n)$, simply because we may have more than $\poly(K \log n)$ such nodes.  The key observation to achieve the claimed $\poly(K \log n)$ decoding time is that most of the remaining nodes (except $\poly(K \log n)$ ones) form at most $\poly(K \log n)$ periodic substrings with periods of lengths at most $\poly(K \log n)$. We thus can make use of suffix trees and longest common prefix queries to speed up the computation of the optimal alignment for the nodes in those periodic substrings. The full proof is given in Appendix~\ref{appendix:proof_decode}.
\end{proof}

\input{analysis}

%% file: analysis.tex
\subsection{Proof of Lemma~\ref{lem:anchor}}
\label{sec:anchor}

\paragraph{Proof Idea.}  We say an alignment $\O$ passes (or goes through) a pair $(u,v)$ if $(u,v)$ is an edge in $\O$.  We choose a particular optimal alignment $\O$ we call the {\em greedy} alignment, and show that $\O$ passes all anchors with high probability.  The high level idea is that suppose on the contrary that $\O$ does not pass an anchor $(u,v)$, then we can find a matching $\M$ in the left neighborhood of $(u,v)$ which may ``mislead'' a random walk, that is, with a non-trivial probability the random walk will ``follow'' $\M$ and consequently miss $(u,v)$. We then have that with high probability at least one of the $\rho$ walks will miss $(u,v)$, which means that $(u,v)$ is not an anchor. A contradiction.

We now give some ideas on how to construct $\M$ and then show that it will mislead a random walk.  Define the left neighborhood of $(u,v)$ that we are interested to be $(s[u-z..u], t[v-z..v])$ such that $s[u] = t[v], \ldots, s[u-z] = t[v-z]$  but $s[u-z-1] \neq t[v-z-1]$.  By exploring the properties of the greedy alignment $\O$, we can find a matching $\M \subseteq \O$ in the left neighborhood of $(u,v)$ consisting of a set of clusters (consecutive edges), each of which is periodic\footnote{That is, the two substrings of the cluster in $s$ and $t$ are both periodic, and of the same period length.} with a small period. Moreover, there  are a small number of singleton nodes between those clusters. Our key observation is that once a random walk $\W$ enters a cluster, its shift (i.e., $\abs{p - q}$ for a walk state $(p, q)$) will be changed by at most the length of the period of that cluster.  We thus can show that the maximum shift change of $\W$ after entering the left neighborhood of $(u,v)$ can be bounded by roughly $k^2\ (k = ed(s,t))$. Now if $\O$ does not pass $(u,v)$, then we can show that with a constant probability the first state $(p,q)$ of $\W$ after entering the neighborhood does not align with $(u,v)$ (i.e. $\abs{p - q} \neq \abs{u - v}$). And then by Lemma~\ref{lem:gambler}, we have with probability at least (roughly) $1/(100k^2)$ that during the whole walk in the neighborhood, the shift of $\W$ will not be equal to $\abs{u-v}$, and consequently $\W$ will miss $(u,v)$.  

\vspace{-2mm}
\paragraph{\bf The Full Proof.}
We first consider anchors $(u,v)$ with $\abs{u -v} \le c_s k$ where $c_s$ is a large enough constant.

Define the {\em greedy} alignment between $s$ and $t$, denoted by $\O$, as follows.  Among all optimal alignments between $s$ and $t$, $\O$ is the one that {\em minimizes} the sum of indices of all matched nodes in $s$ and $t$, breaking ties arbitrarily.  We will show that with high probability this particular optimal alignment passes all anchors.  As mentioned earlier, an alignment between the two strings naturally corresponds to a non-crossing matching between the characters (also called nodes) of the two strings. Thus for convenience we also use the notation $\O$ for the corresponding matching between $s$ and $t$, and let $\abs{\O}$ be the number of edges in the matching $\O$.

Consider a particular pair $(u,v)$.  Let $z \ge 0$ be a value such that 
\begin{equation*}
s[u] = t[v], \ldots, s[u-z] = t[v-z], \text{ but } s[u-z-1] \neq t[v-z-1].
\end{equation*}
We can assume that either $u-z-1 \ge 1$ or $v-z-1 \ge 1$ since otherwise if $u-z = v-z = 1$ then $\O$ will just follow $(1,1), (2,2), \ldots, (u,v)$, and consequently pass $(u,v)$ with certainty.  From now on we assume that $\O$ does not pass $(u,v)$ since otherwise we are done. We will show that a random walk starting at $(1,1)$ will miss $(u,v)$ with some non-trivial probability. 

\begin{figure}[t]
\centering
\includegraphics[height = 1.4in]{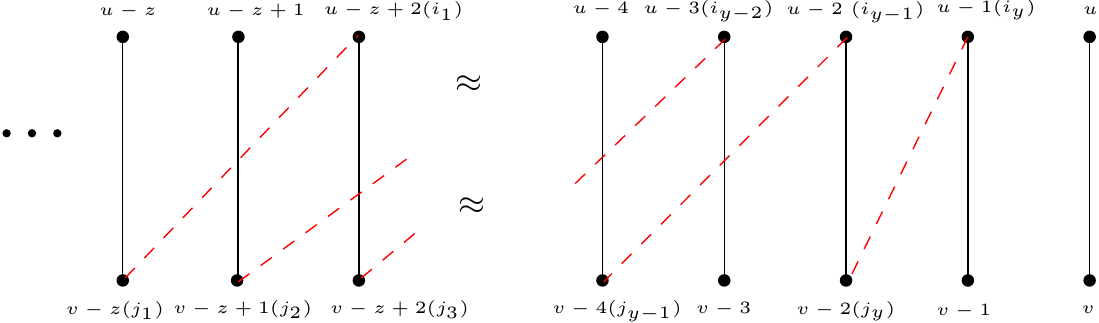}
\caption{The matching $\M$ (red dashed edges) in the greedy optimal alignment $\O$. Black solid edges are those in $\S$.}
\label{fig:matching}
\end{figure}

Let $\S = \{(u-z, v-z), \ldots, (u,v)\}$ be a matching in the left neighborhood of $(u,v)$. We call $s[u-z .. u]$ and $t[v-z..v]$ the {\em stable} zone.  Let
$$
\M = \{(i,j) \in \O\ |\ (u-z \le i \le u) \wedge (v-z \le j \le v)\}
$$
be a subset of the greedy matching of $\O$ in the stable zone.  Let $(i_1, j_1), \ldots, (i_y, j_y)$ be the set of edges in $\M$,  sorted increasingly according to the indices of the ends of the edges.  For a pair $(i,j)\ (i, j \in [n])$, let $d(i,j) = (v-j) - (u-i)$ be the (signed) {\em shift} of $(i,j)$ from $(u,v)$. W.l.o.g.\ we can assume that $d(i_y, j_y) > 0$, since the case $d(i_y, j_y) < 0$ is symmetric, and if $d(i_y, j_y) = 0$ then the greedy matching $\O$ must include $(i_y, j_y), (i_y+1, j_y+1), \ldots$, and consequently include $(u,v)$; again we are done. We refer readers to Figure~\ref{fig:matching} for an illustration of $\M$.  We have the following observations.

\begin{figure}[t]
\centering
\includegraphics[height = 1.1in]{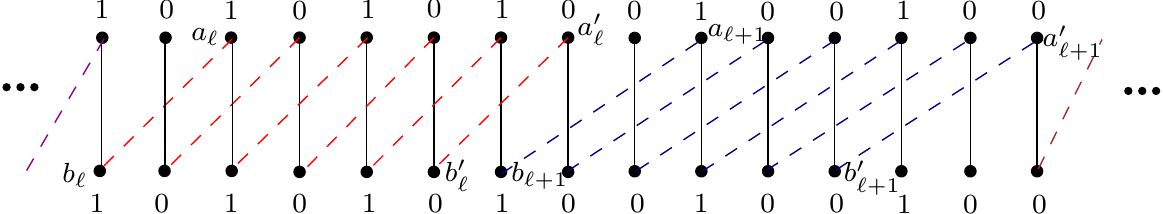}
\caption{(Color dashed) edge-clusters in $\M$. Black solid edges are those in $\S$.}
\label{fig:step2}
\end{figure}

\begin{claim}
\label{cla:cluster}

Given that $d(i_y, j_y) > 0$, the matchings $\M$ and $\O$ have the following properties.
\begin{enumerate}
\item For any edge $(i,j) \in \M, \ d(i, j) >0$.

\item For any edge $(i,j) \in \O$, $\abs{d(i,j)} \le c_d k$ for a large enough constant $c_d$.

\item The edges in $\M$ form $m \le k+1$ clusters, where within each cluster the shifts of all matching edges are the same. Let $\C_1, \ldots, \C_m$ denote the $m$ clusters from left to right.   Let $(a_\ell, b_\ell)$ and $(a'_\ell, b'_\ell)$ be the first and last edges in $\C_\ell$. (see Figure~\ref{fig:step2} for an illustration.) We have that each $\C_\ell$ contains two periodic strings in $s$ and $t$ with the form $\varphi_\ell \cdots \varphi_\ell \varphi'_\ell$ where $\varphi_\ell$ is the period with $\abs{\varphi_\ell} \le c_d k$, and $\varphi'_\ell$ is a prefix of $\varphi_\ell$ (can be $\emptyset$).   Call $\varphi_\ell$ the period of $\C_\ell$.
	
\item The number of unmatched nodes between $s$ and $t$ in $\O$ is at most $2k$.  
\end{enumerate}
\end{claim}
  
\begin{proof}
Let $d = v-u$.  For Item $1$, suppose this is not the case, in other words, there exists some $x \in [y-1]$ such that $d(i_{x}, j_{x}) \le 0$ while $d(i_{x'}, j_{x'}) > 0$ for $\forall x' \in (x, y]$.  We can replace the set of matching edges $(i_{x+1}, j_{x+1}), \ldots, (i_y, j_y)$ in $\O$ with edges $(i_{x+1}, i_{x+1}+ d), \ldots, (i_{x+1}+y-(x+1), i_{x+1}+y-(x+1)+d)$ in $\S$, getting a new matching $\O'$.  It is easy to see that after such a replacement we either have $\abs{\O'} > \abs{\O}$, or $\abs{\O'} = \abs{\O}$ but $\O'$ has a smaller sum of indices of all matching nodes in $s$ and $t$ compared with $\O$; both cases contradict the fact that $\O$ is a greedy matching.

Item $2$ is an application of the triangle inequality on $\abs{i - j} \le k$ and $\abs{u-v} \le c_s k$, where the former is because $(i, j)$ is in an optimal alignment, and the latter is due to the type of pair $(u,v)$ that we are currently considering.

For Item $3$, we have at most $k+1$ clusters since $ed(s,t) = k$.  For any $(i,j) \in \C_\ell$, we have $s[i] = t[j]$; and we also have $s[i] = t[i+d]$ since $(i, i+d) \in \S$. We thus have $t[j] = t[i+d]$ by transitivity. Then $\abs{\varphi_\ell} = (i+d) - j = d(i,j) \le c_d k$.

Item $4$ is obvious since $ed(s,t) \le k$ (the constant $2$ is due to the substitutions). 
\end{proof}

We now consider a random walk $\W$.  Since $s[u-z-1] \neq t[v-z-1]$ by the definition of $z$, we have with probability at least $2/3$ that the random walk must pass $(u-z, v-z-\alpha)$ or $(u-z-\alpha, v-z)$ where $\alpha \ge 1$.  
W.l.o.g.\ we can assume the former, since in the latter case we can just consider a mirroring matching $\M'$ of $\M$ by ``flipping'' all the edges in $\M$, that is, for any edge $(u - \beta, v - \gamma)$ in $\M$ we add a mirroring edge $(u - \gamma, v - \beta)$ to $\M'$. Note that $s[u - \gamma]$ and $t[v - \beta]$ can indeed be matched since if $(u - \beta, v - \gamma)$ is an edge in $\M$ then $s[u - \beta] = t[v - \gamma]$, and moreover we have $s[u - \beta] = t[v - \beta]$ and $s[u - \gamma] = t[v - \gamma]$ since both $(u-\beta, v-\beta)$ and $(u-\gamma, v-\gamma)$ are edges in $\S$, and then  by transitivity we must have $s[u - \gamma] = t[v - \beta]$.

Let $d(p,q) = (u-p) - (v-q)$ denote the (signed) {\em shift} of a walk state $(p, q)$ from $(u,v)$.   We consider the first state $(p_0, q_0)$ and the last state $(p_1, q_1)$ of $\W$ that fall into the stable zone (i.e., $u-z \le p_0, p_1 \le u$ and $v-z \le q_0, q_1 \le v$).  We will consider the walk states between (inclusive)  $(p_0, q_0)$ and $(p_1, q_1)$, and call $s[p_0..p_1]$ and $t[q_0..q_1]$ the {\em confusing} zone. Note that if we can show that the shift of $\W$ is never equal to $0$ in this confusing zone, then $\W$ will miss $(u,v)$.  We will show that this happens with some non-trivial probability if $\O$ does not pass $(u,v)$.

We first show some properties of the two boundaries $(p_0, q_0)$ and $(p_1, q_1)$.  For convenience we define the number of nodes in $s[x..y]$ to be $0$ if $y < x$.
\begin{claim}
\label{cla:boundary}
Considering a random walk $\W$ starting from $(u-z, v-z-\alpha)$ with $\alpha \ge 1$, we have:
\begin{enumerate}
\item With probability at least $1/3$, $d(p_0, q_0) \ge 1$ and $q_0 = v - z$.    

\item The total number of nodes in $s[p_0..a_1]$ and $t[q_0..b_1]$ is no more than $2c_d k$, and the total number of nodes in $s[a_m..p_1]$ and $t[b_m..q_1]$ is no more than $2c_d k$.
\end{enumerate}
\end{claim}

\begin{proof}
For Item $1$, since $\W$ starts from $(u-z, v-z-\alpha)$, if $s[u-z] = t[v-z-\alpha]$, then the next walk state of $\W$ will be $((u-z)+1, (v-z-\alpha)+1)$, and after a few more walk steps, the first walk state $(p_0, q_0)$ that fully falls into the stable zone must have the property that $d(p_0, q_0) \ge 1$ and $q_0 = v - z$.  Otherwise if $s[u-z] \neq t[v-z-\alpha]$, then with probability $1/3$, after the first progress step the state of $\W$ will be $((u-z)+1, (v-z-\alpha))$, and then the same argument apply. 

To show the first inequality of Item $2$, let $(a'_0, b'_0)$ be first edge in $\O$ to the left of $(a_1, b_1)$.  It must be the case that (1) $d(a'_0, b_1) = 0$ (thus $a'_0$ is in the stable zone) and (2) $b'_0 < v - z$.  The former is true since otherwise we can replace $(a_1, b_1)$ in $\O$ with $(a'_0, b_1)$ to obtain another alignment $\O'$ with a smaller sum of indices of all matching nodes in $s$ and $t$.  The latter is true since $(a_1, b_1)$ is the first edge in $\O$ in the stable zone (otherwise if $b'_0 \ge v - z$ then $(a'_0, b'_0)$ will be the first edge).  Now we have 
\begin{enumerate}
\vspace{-1mm}
\item $0 < d(a'_0, b'_0) \le d(a'_0, b_1)  +  \abs{b_1 - b'_0} \le 0 + k = k$, where the second term $k$ counts the number of unmatched nodes in $\O$ between $t[b'_0]$ and $t[b_1]$.

\vspace{-1mm}
\item $0 < d(a_1, b_1) \le c_d k$, by Item $2$ of Claim~\ref{cla:cluster}.

\vspace{-1mm}
\item $\abs{a'_0 - (u-z)} \le \abs{b_1 - b'_0} \le k$. 

\vspace{-1mm}
\item $b'_0 < v - z$.
\end{enumerate}
\vspace{-1mm}
These inequalities imply that the total number of nodes in $s[u-z..a_1]$ and $t[v-z..b_1]$ is at most $c_d k + 2k \le 2c_d k$, and consequently the total number of nodes in $s[p_0..a_1]$ and $t[q_0..b_1]$ is at most $2c_d k$ since the walk state $(p_0, q_0)$ is inside the stable zone. A similar proof applies to the second inequality.
\end{proof}

We now bound the maximum change of the shift of the random walk $\W$ in the confusing zone. The key observation is that when the random walk travels through the cluster $\C_\ell$ (i.e., the set of walk states $(p,q)$ with $a_\ell \le p \le a'_\ell$ and $b_\ell \le q \le b'_\ell$), the maximum change of $\W$'s shift is upper bounded by the length of $\C_\ell$'s period $\varphi_\ell$, since the shift will stop changing as soon as the walk reaches a state $(p, q)$ where the difference between the shift of $(p, q)$ and the shift of the first edge of $\C_\ell$ (i.e., $(a_\ell, b_\ell)$) is a multiple of $\abs{\varphi_\ell}$.  We now try to bound the maximum change of $\W$'s shift in the rest of the walk steps in the confusing zone.   By Item $4$ of Claim~\ref{cla:cluster}, we have at most $2k$ nodes in the gaps of the $m$ clusters $\C_1, \ldots, \C_m$; call these nodes $G_1$. And by Item $2$ of Claim~\ref{cla:boundary}, we have at most $4c_d k$ nodes in the gaps between $(p_0,q_0)$ and $\C_1$ and between $\C_m$ and $(p_1,q_1)$ (i.e., in the two ends of the confusing zone); call these nodes $G_2$.  By a Chernoff bound, with probability $1 - 1/k^3$, $\W$ will have at most $(\abs{G_1} + \abs{G_2}) \cdot 4\log k \le 20c_d k \log k$ walk states $(p,q)$ where $p \in G_1 \cup G_2$ or $q \in G_1 \cup G_2$. Therefore with probability $1-1/k^3$, the change of shift of $\W$ inside the confusing zone but outside the $m$ clusters is bounded by $20c_d k \log k$. Summing up, with probability $1-1/k^3$, the maximum change of the shift made to $\W$ in the confusing zone (before it misses $(u,v)$) is bounded by 
\begin{equation}
\label{eq:z-1}
\textstyle 20c_d k \log k + \sum_{\ell \in [m]} \abs{\varphi_\ell} < 2c_d k^2.
\end{equation}

Recall that with probability $2/3 \cdot 1/3 = 2/9$, the initial shift of $\W$ in the confusing zone is $d(p_0, q_0) \ge 1$ (Item $1$ of Claim~\ref{cla:boundary}), conditioned on which, by Lemma~\ref{lem:gambler} we have that with probability $1/(4c_d k^2)$ the shift of $\W$ reaches $2c_d k^2$ before $0$. Then by (\ref{eq:z-1}) and a union bound, we conclude that $\W$ will miss $(u,v)$ with probability at least $1 - 1/k^3$.  Thus (if $\O$ does not pass $(u, v)$ then) $\W$ will miss $(u,v)$ with probability at least 
\begin{equation*}
\label{eq:miss-prob}
2/9 \cdot {1}/{(4c_d k^2)} - {1}/{k^3} \ge 1/(20c_d k^2). 
\end{equation*}
Thus if we have $\rho = c_\rho K^2 \log n$ (for a large enough constant $c_\rho$) random walks, the probability that at least one of the $\rho$ walks will miss $(u,v)$ is at least 
$$1 - \left(1 - 1/(20c_d k^2) \right)^{c_\rho K^2 \log n} \ge 1 - 1/n^4.$$   
In other words, if all the $\rho$ random walks pass $(u,v)$ (so that $(u,v)$ is an anchor), then with probability $(1 - 1/n^4)$ the pair $(u,v)$ is included in  $\O$.  By a union bound on at most $n^2$ possible pairs $(u,v)$, we conclude that with probability $(1 - 1/n^2)$ the greedy matching  $\O$ goes through all anchors $(u,v)$ for which $\abs{u-v} \le c_s k$.
\smallskip

We now consider those pairs $(u,v)$ for which $\abs{u-v} > c_s k$.  We will show that with very high probability at least one of the $\rho$ random walks will miss $(u,v)$, and consequently $(u,v)$ is not an anchor.  To see this, consider a random walk $\W$. By Lemma~\ref{lem:CGK} we know that with probability at least $0.9$, the number of progress steps in $\W$ is at most $c_n k^2$ for a large enough constant $c_n$, conditioned on which, by Lemma~\ref{lem:shift}, we have that with probability $0.9$, the shifts of $\W$ will never be more than
$c'_s k$ for a large enough constant $c'_s$. Thus if we set $c_s = 2c'_s$, then with probability $0.9 \cdot 0.9 > 0.8$, $\W$ will miss the pair $(u,v)$. Therefore the probability that at least one of the $\rho$ random walks will miss $(u,v)$ is at least $1 - (1 - 0.8)^{\rho} \ge 1 - 1/2^{\Omega(K^2 \log n)}$. By a union bound on at most $n^2$ such pairs, we conclude that with probability $1 - 1/2^{\Omega(K^2 \log n)}$ (thus the failure probability is negligible), all pairs $(u,v)$ with $\abs{u -v} > c_s k$ will {\em not} be anchors.  This concludes the proof of the lemma.

%% file: streaming.tex
\section{Streaming}
\label{sec:stream}

In this section we give algorithms for both the simultaneous streaming model in which we can scan the two strings $s$ and $t$ simultaneously in the coordinated fashion, and the standard streaming model in which we can only scan them one by one in one pass.

\subsection{Simultaneous Streaming}
\label{sec:sim-stream}
We show the following theorem for the simultaneous streaming model.
\begin{theorem}
\label{thm:sim-stream}

There exists an algorithm in the simultaneous streaming model that computes the edit distance using $O(K)$ words of space and $O(n + K^2)$ processing time, where $n$ is the input string size and $K$ is the distance upper bound.
The algorithm can be extended to compute all edit operations at the expense of increasing the space to $O(K^2)$ words. 
\end{theorem}

\begin{figure}[t]
\centering
\includegraphics[height = 1.8in]{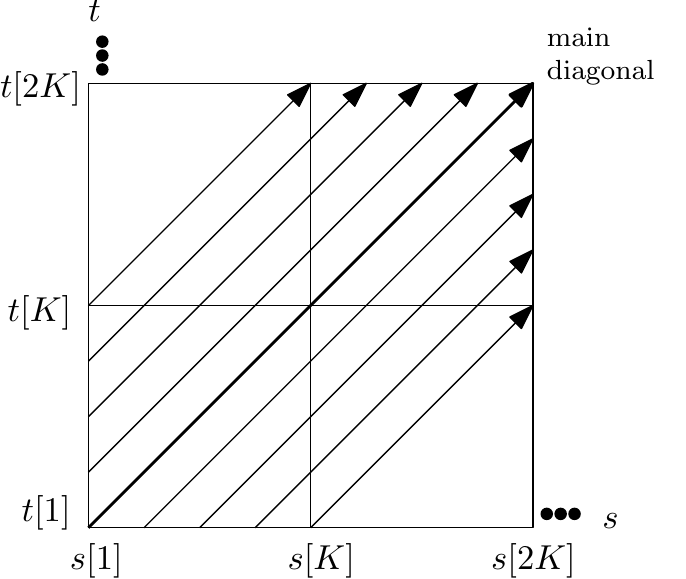}
\caption{Dynamic Programming in the Streaming Model}
\label{fig:DP}
\end{figure}

Our algorithm adopts some ideas from the $O(n + K^2)$ algorithm for computing edit distance in the RAM model~\cite{LMS98}.  The main idea is to perform the dynamic programming in a space efficient way along the $2K+1$ diagonals around the main diagonal in the alignment matrix, which is an $n \times n$ matrix where the $(i,j)$-th cell contains the $ed(s[1..i], t[1..j])$ (called the {\em score} of that cell). See Figure~\ref{fig:DP} for an illustration. 
It is easy to see that the scores in each diagonal from bottom-left to top-right are non-decreasing.  We will show that for computing the edit distance between $s$ and $t$ (or reporting ``error'' if $ed(s,t) > K$), we only need to store $O(1)$ scores in each of the $2K+1$ diagonals at any moment.

We will make use of the suffix tree which, once built, allows to compute the longest common prefix between suffixes of two strings in $O(1)$ time. The suffix tree  can be built in linear time \cite{weiner1973linear}.

\vspace{-2mm}

\paragraph{The Algorithm.} 
Our algorithm runs in $n/K$ phases.  At each phase $i \in [n/K]$ we run over each of the $2K+1$ diagonals towards top-right (see the arrows in Figure~\ref{fig:DP}), up to the first cell where it intersects the $(iK)$-th row  (counting bottom up) or $(iK)$-th column  (counting from left to right); see Figure~\ref{fig:DP}.  We call the $(iK)$-th row and $(iK)$-th column the {\em boundary} of phase $i$.  Each phase $i$ starts with a preprocessing step, in which we read two substrings $s_i = s[\min\{1, (i-2)K\}.. iK]$, $t_i = t[\min\{1, (i-2)K\}.. iK]$, and build a suffix tree which allows to answer the longest common prefix of suffixes of $s_i$ and $t_i$.  During the execution of the phase we try to maintain for each of the $2K+1$ diagonals the following information:
\begin{itemize}
\vspace{-1mm}
\item {\em Score of the diagonal}: the score in the highest cell reached in the diagonal, denoted by $L$.
\vspace{-1mm}
\item {\em Boundary flag}: a bit indicating whether the highest reached cell is on the boundary of the current phase.  Initialized to be $0$ at the beginning of each phase.
\vspace{-1mm}
\item {\em Last change}: The cell on the diagonal at which the score changes from $L-1$
to $L$. 
\vspace{-1mm}
\item {\em Second to the last change}: The cell on the diagonal at which the score changes from $L-2$
to $L-1$. 
\end{itemize}
\vspace{-1mm}
We also maintain $K+1$ lists (list $0$ to list $K$), each consisting of a subset of the $2K+1$ diagonals. More precisely, the $\ell$-th list contains all diagonals whose scores are equal to $\ell$. At the beginning of the first phase all diagonals are in list $0$.  At each phase we start by processing all diagonals in list $0$, then all diagonals in $1$, and so on.  Each time a diagonal in list $\ell$ is processed, it either moves to list $\ell+1$ if its score increases to $\ell+1$, or stays in list $\ell$ if the run of that diagonal hits the phase boundary and the score of the diagonal stays equal to $\ell$. When the score of a diagonal reaches $K+1$ we drop it from further consideration.  During the computation we have the following invariant.
\vspace{-1mm}
\begin{quote}
{\em Invariant: The scores of any two neighboring diagonals cannot differ by more than $1$}. 
\end{quote}

We now show how to update the scores of diagonals in each phase $i$. When we try to process a diagonal $D$ in list $\ell$, by the invariant the scores of its two neighboring diagonals must in $\{\ell-1, \ell, \ell+1\}$. If a neighboring diagonal has score $\ell-1$, then it must be the case that the run of that diagonal already hits the boundary of the current phase.  In this case by the invariant we know immediately that we can extend the run of $D$ to the boundary of the current phase, and the score of $D$ stays equal to $\ell$.  Otherwise the scores of the two neighboring diagonals are in $\{\ell, \ell+1\}$, in which case we extend the run of $D$ using the following information kept in the memory:
\begin{itemize}
\vspace{-1mm}
\item The longest common prefix of the two strings $s[p, iK]$ and $t[q, iK]$, where $(p, q)$ is the current cell reached by the run on $D$.  The information can be obtained in $O(1)$ time by querying the suffix tree. 
\vspace{-1mm}
\item The cell on the left neighboring diagonal at which the score changes from $\ell-1$ to $\ell$. 
\vspace{-1mm}
\item The cell on the right neighboring diagonal at which the score changes from $\ell-1$ to $\ell$. 
\end{itemize}
\vspace{-1mm}
Each of the three pieces of information will give a candidate cell on $D$ up to which the run can extend, and we then take the highest among the three candidates  (this is essentially the same as the taking-maximum step in the standard dynamic programming algorithm for computing edit distance). Now if the run hits the boundary in the middle towards the highest candidate, then we stop the run, set $D$'s boundary flag to $1$, and keep $D$'s score at value $\ell$. Otherwise we update $D$'s score to value $\ell+1$.
\vspace{-2mm}

\paragraph{The Analysis.}  In essence, our algorithm mimics the standard dynamic programming algorithm for computing edit distance (but in a space and time efficient way).   The correctness of our algorithm follows directly from that of the standard dynamic programming.  

For the running time, note that each of the $2K+1$ diagonals will be processed at most $O(n/K + K)$ times, which follows from the fact that at the end of each of the $n/K$ phases, either the run of the diagonal reaches the boundary of the current phase, or the score of the diagonal increments (recall that the diagonal will be dropped once its score reaches $K+1$).  Since each processing takes $O(1)$ time, the total running time is bounded by $O(n + K^2)$.  
For the space, it is easy to see that the algorithm only uses $O(K)$ words since we only store $O(1)$ words of information for each diagonal at any time.  

Finally, the algorithm can be easily modified to reconstruct the sequence of edit operations. We simply notice that the longest successive common prefixes on the same diagonal can be merged so that we do not need more than $K$ of them. We thus can keep the starting and ending positions of those common substrings in the memory, which costs $O(K^2)$ words. At the end of the algorithm we can reconstruct the optimal alignment backwards by looking at the starting and ending positions of those common substrings.

\subsection{Standard Streaming}
\label{sec:standard-stream}

Our algorithm in the standard streaming model follows directly from our sketching algorithm in Section~\ref{sec:sketch}, since the encoding phase in our sketching algorithm can be done in the one-pass streaming model: the CGK embeddings can be performed in the streaming model; the (rolling) hash signatures and lengths of blocks can be computed in the streaming fashion; and the redundancies can also be computed in the streaming fashion by Lemma~\ref{lem:ham}. We thus just need to first sketch string $s$ and keep the sketch in the memory, and then do the same thing for string $t$, and at the end perform the decoding using the two sketches stored in the memory.
 
\begin{theorem}
\label{thm:standard-stream}

There exists an algorithm in the standard streaming model that computes with probability $0.9$ the edit distance and all the edits using $O(K^8 \log^5 n)$ bits of space and $\tilde{O}(K^2 n)$ time, where $n$ is the input string size and $K$ is the distance upper bound.
\end{theorem}

%% file: related.tex
\section{Concluding Remarks}
In this paper we have proposed an improved algorithm for document exchange that matches the information theoretic lower bound $\Omega(K \log n)$ when $\log K = O(\sqrt{\log n})$ while maintaining almost linear encoding/decoding time. We have also proposed the first sketching and streaming algorithms with sketch/space size $\poly(K \log n)$.

Although we have made a significant progress on the three problems related to edit distance, a number of questions remain open. First, for document exchange, can we further improve the communication to optimal bound $O(K \log n)$ for all values $K$ and  $n$, while maintaining  (almost) linear running time?  Second, in the sketching problem, what are the best polynomial dependencies on $K$ and $\log n$ in the sketch size? Can we prove any lower bounds?  In the analysis of our sketching algorithm we did not attempt to optimize the polynomial dependencies. We guess that with a more careful analysis (e.g., replacing some brute force union bounds) of our algorithm, the dependency on $K$ can be reduced to $K^4$ or even $K^3$, but what is the best that we can achieve? Finally, is it possible to derandomize our algorithm for document exchange to obtain a better error-correcting code for edit distance?
 
\section*{Acknowledgements} The authors would like to thank Hossein Jowhari for helpful discussions at the early stage of this work. The second author would like to thank Funda Ergun, Cenk Sahinalp, Dirk Van Gucht and Erfan Sadeqi Azer for helpful discussions and comments.

%% file: appendix.tex

\section{Missing Proofs in the Preliminaries}

\subsection{Proof of Theorem~\ref{thm:IMS} (The Improved IMS Algorithm)}
\label{proof:thm:IMS}

We describe our improved IMS algorithm and the analysis below.

\vspace{-2mm}
\paragraph{Encoding.}
The scheme uses a pair of pairwise independent hash functions $f_1: [n] \to [(K\log n)^{c}]$ and $f_2: [n]^* \to [(K\log n)^{c}]$.  In words, the function $f_2$ is a rolling hash function (e.g., the Karp-Rabin hashing) that map substrings to $[(K\log n)^{c}]$.  The encoding process is divided into levels $\ell = 1, 2, \ldots$. At the top level ($\ell = 1$), Alice divides her string $s$ into $2K$ blocks (substrings) $x_1, \ldots, x_{2K}$ each of length $n/(2K)$, and then sends to Bob a vector $U_1$ of length $2K$ where the $i$-th coordinate is set to be $f_2(x_i)$.   

At the level $\ell \ge 2$, Alice creates a {\em signature} vector $U_\ell [1..(K\log n)^{c}]$ with each coordinate being an integer in $[(K\log n)^{c}]$ (encoded using $c(\log K+\log\log n)$ bits) initialized to $0$.  Alice cuts each block at level $\ell - 1$ into two (approximately) equal length sub-blocks, that is, one of length $\lfloor b/2 \rfloor$ and the other of length $\lceil b/2 \rceil$ if the original block is of length $b$. Then for a block number $i$ containing substring $x$, Alice sets $U_\ell[f_1(i)] = U_\ell[f_1(i)] \XOR f_2(x)$ (here $\XOR$ denotes the bit-wise exclusive-or operator). Alice then uses the scheme in Lemma~\ref{lem:ECC} (setting $p = 1/(K \log n)^{c_p}$ for a large enough constant $c_p$) to compute a redundancy of $U_\ell$ which allows to recover up to $2K$ errors.  Alice continues this process until reaching a level $L$ at which all blocks are of length at most $c(\log K+\log\log n)$ bits. At this level, for the $i$-th block $x$,  Alice sets $U_L (f_1(i)) = U_L[f_1(i)] \XOR x'$, where $x'$ is equal to $x$ padded with zeros if it is of length less than $c(\log K+\log\log n)$, and then computes a redundancy of $U_L$ that allows to recover up to $2K$ errors.  Finally, Alice sends Bob the redundancies of $U_2, \ldots, U_L$ computed at all levels. 

\vspace{-2mm}
\paragraph{Decoding.}
The decoding phase also works by levels.  At each level Bob tries to reconstruct $K$ blocks of Alice's input $s$ by matching their hash signatures against that of the substrings of his input $t$.  At the top level ($\ell = 1$), Bob tries to match every block of Alice's input $s$ against  the $2K+1$ substrings of his input $t$ using Alice's message $U_1$: If the block starts at position $j$ in $s$, then Bob compares its hash signature against the signatures of substrings of $t$ of same length starting at positions ranging from $j-K$ to $j+K$.  In this way at most $K$ blocks will not find a match. Bob then copies the contents of the matching substrings from $t$ into the corresponding blocks in $s$. There will be at most $K$  blocks in $s$ left unmatched at the top level. 

At the level $\ell \ge 2$, Bob creates a vector $V_\ell$ of the same size as $U_\ell$ initialized to $0$.  He then divides each matched block into two sub-blocks, and for each sub-block $x$ that matches the $i$-th block of $s$ at the same level,  sets $V_\ell[f_1(i)] = V_\ell[f_1(i)] \XOR f_2(x)$ (here $\XOR$ is applied coordinate-wise). Bob then uses the redundancy of $U_\ell$ sent by Alice to recover $U_\ell$ from $V_\ell$, and computes the vector $W_\ell = U_\ell \XOR V_\ell$. He can now recover the signatures of $2K$ missing blocks in $s$ as follows. For a block at position $i$, he copies its signature from the entry $W_\ell[f_1(i)]$. He then compares again the signature of each missing block with that of the $2K+1$ substrings of his input $t$ as before, and tries to match at least $K$ of the $2K$ missing blocks. Bob continues in the same way until reaching the bottom level $L$, at which the copied entries will be the actual contents of the $2K$ missing blocks. 

\vspace{-2mm}
\paragraph{Analysis.}
We now prove the correctness of the algorithm and analyze its costs. It is easy to see that at each level $\ell$, the vector $W_\ell$ only contains entries of signatures corresponding to the missing blocks. Moreover, since we have at most $2K$ missing block signatures and the size of vector $W_\ell$ is $(K\log n)^c$, the probability of not having any collision is at least $1 - (2K)^2/(K\log n)^c = 1 - 4/(K^{c-2}\log^{c} n)$.  Applying a union bound on $\log n$ levels, the overall success probability is at least $1- 1/(K \log n)^{\Theta(1)}$.

The communication cost at each level is
$O(\log (2K) + K (\log(K \log n) + \log K + \log(K \log n))) = K (\log K + \log\log n).$
Summing up over all the $\log n$ levels the cost is $O(K(\log K + \log\log n) \log n)$. 
 The preprocessing for computing all $f(x)$ for substrings of $s$ and $t$ can be done in $O(n)$ time using the Rabin-Karp hashing.  At each level the decoding and encoding time of the scheme in Lemma~\ref{lem:ECC} takes $\tilde{O}(n)$.  Filling vectors $U_\ell$ and $V_\ell$ takes time $O(2^\ell K)$ at each level $\ell$, until the bottom level $L$ at which it takes $O(n/(\log K + \log \log n))=o(n)$ time. The only significant steps that remains to analyze is the comparisons of string signatures.  At each level we have $2K$ block signatures, where each signature is compared against $2K + 1$ substrings.  Thus at each level the cost is $O(K^2)$, and summing over all the $\log n$ levels we get $O(K^2 \log n) = O(n)$.

\subsection{Proof of Lemma~\ref{lem:walk-break}}
\label{proof:lem:walk-break}

\begin{proof}
Since $w$ contains no substring of length at least $\ell$ with period at most $\theta$, it is easy to see that two substrings of $w$ of length $\ell$ starting at two positions $z_0, z_1$ with $\abs{z_0 - z_1} \in [1..\theta]$ can not be equal. In particular, that means that if the random walk starts from the state $(p,q)$ with $|(p-i)-(q-j)|\in[1..\theta]$, $p\in [i..i+m-\ell]$ and $q\in [j..j+m-\ell]$, then there must be a progress step before it reaches the state $(p+\ell,q+\ell)$ since the substrings $s[p..p+\ell-1]$ and $t[q..q+\ell-1]$ are not identical. 
\end{proof}

\subsection{Proof of Lemma~\ref{lem:ECC}}
\label{proof:lem:ECC}

\begin{proof}
We first consider a special case where $\lambda = k$.  We use the universal hashing proposed in \cite{carter1977universal} together with an error-correcting code 
(e.g., one can use the Reed-Solomon code with the decoding algorithm by Gao~\cite{Gao03}). Let $f : [u] \to [v]\ (v=k^2/p)$ be a hash function chosen randomly from a $(2,2)$-universal family of hash functions~\cite{dietzfelbinger1992polynomial}. 

Alice first creates a vector $V[1..v]$ initialized to $0$, and for each $i \in [u]$, she sets $V[f(i)] = V[f(i)] \oplus a[i]$.  She then computes a redundancy that allows to correct up to $k$ errors on $V$ using a systematic Reed-Solomon code, and sends it to Bob. The size of the redundancy is $O(k(\log v+\log \sigma))=O(k(\log\sigma+\log k+\log (1/p)))$, and the size of the description of the hash function $f$ is $O(\log v+\log\log u)$.

Bob can now recover the vector $a$ as follows. He first computes an array $V'$ (initialized to $0$) as follows: for all $i \in [u]$ that are not in the $k$ coordinates where $a$ and $b$ may differ (recall that Bob knows the indices of these coordinates), he sets $V'[f(i)] = V'[f(i)] \oplus b[i]$. He then uses the redundancy received from Alice to recover $V$ from $V'$, and from there computes $W = V \oplus V'$. Bob outputs ``error'' if the error-correcting step (i.e., recovering $V$) fails. Finally for all $i$ in the $k$ coordinates where $a$ and $b$ may differ, Bob sets $a[i]=W[f(i)]$. 

We now show that the recovery succeeds with probability $1-p$. To see this, notice that since we have removed all the $(u - k)$ indices at which $a$ and $b$ must agree by computing $W = V \oplus V'$, the resulting information in $W$ are those $k$ coordinates $a[i]$ where it is possible that $a[i] \neq b[i]$. We can easily recover those $k$ coordinates if there is no collision, that is, no pair of $(a[i], a[j]) \ (i \neq j)$ among those $k$ coordinates such that $f(i) = f(j)$, which holds with probability at least $1 - p$ by using the universal hash function $f$.
The running time of the algorithm is linear except for the encoding/decoding of the error-correcting which takes time $O(v \ \mathtt{polylog} (u))$~\cite{Gao03}. 

We now consider the general $\lambda \ge k$.  The algorithm is similar to the special case above, with a few modifications: (1) we set $v = 4k\lambda/p$; and (2) once Bob finds out the set of (at most) $k$ coordinates (denoted by $X$) where $V$ and $V'$ differ after the error-correcting step, for each $a[i]$ in the $\lambda$ coordinates where errors could occur, he checks whether $f(i) \in X$, and if so sets $a[i] = W[f(i)]$.  Now the algorithm can fail in two ways. The first is again due to collisions, whose probability is upper bounded by $2k^2/v \le p/2$.  The second is due to false positives, that is, there exists a coordinate $b[i]$ among the $\lambda$ candidates such that $a[i] = b[i]$ and $f(i) \in X$, which happens with probability at most $2\lambda/v \le p/2$.  Thus the total error probability of the algorithm is at most $p$.  The communication cost and running time can be computed in the same way as the special case $\lambda = k$.
\end{proof}

\section{Missing Proofs in Document Exchange}

\subsection{Proof of Lemma~\ref{lem:period_elim} (Periods Elimination)}
\label{proof:lem:period_elim}

\begin{proof}
Let $k = ed(s,t)$.
The proof is via LCS (the longest common subsequence). We first consider a simple version of the edit distance where we only allow insertions and deletions. We will consider later the standard version where substitutions are allowed.  Given two strings $s$ and $t$ of lengths $m$ and $n$ respectively, with edit distance $k$ and LCS $z$, it is easy to see that 
\begin{equation}
\label{eq:a-1}
m+n=k+2z.
\end{equation}  
Let us build a bipartite graph, in which nodes on the left side correspond to characters of $s$ and nodes on the right side correspond to characters of $t$.  The LCS simply corresponds to the bipartite graph (denoted by $G$) with the largest number of edges such that the following two properties hold.
\begin{enumerate}
\item For any edge $(i,j)$, we have that $s[i]=t[j]$. 
\item (\em Non-crossing) For any two edges $(i,j)$
and $(i',j')$, we have that $i\neq i'$, $j\neq j'$
and moreover $i<i'$ iff $j<j'$. 
\end{enumerate}

What we will show is that 
\begin{equation}
\label{eq:a-2}
\text{LCS}(s',t') \ge \text{LCS}(s,t) - \pi.
\end{equation}
Together with (\ref{eq:a-1}) we immediately have $ed(s', t') \le ed(s,t)$, since $s'$ and $t'$ are generated by removing the central substrings of lengths $\pi$ from $s$ and $t$. To show (\ref{eq:a-2}) we just need to show that we can build a bipartite graph $G'$ corresponding to an alignment between the strings $s'$ and $t'$ that has at most $\pi$ edges fewer than $G$. 

We denote $s$ by $p_1p_2p_3\ (p_1 = p_2 = p_3 = p)$, and $t$ by $q_1q_2q_3$; each $q_i$ or $p_i$ is of length $\pi$.  To construct $G'$, we keep all edges in $G$ that do not connect to nodes (characters) in $p_2$ or $q_2$. Thus the only edges we could miss are those that either connect a node in $p_2$ with a node in $q_2$, or a node in $p_2$ with a node in $q_i$ with $i \neq 2$, or a node in $q_2$ with a node in $p_i$ with $i\neq 2$. Let us consider all edges connected with nodes in $q_2$ (sorted in the increasing 
order of the node they connect in $q_2$). We can only have five cases: 
\begin{enumerate}
\item No such edge exists.  This is an easy case: all edges that we can lose are those connected with $p_2$, and thus the number of lost edges cannot be more than $\pi$.  

\item Both the first edge and the last edge are connected with $p_2$.  In this case, all lost edges are connected with $p_2$, and thus the number of lost edges cannot be more than $\pi$. 

\item The first edge is connected with $p_1$ and the last one is connected with $p_3$. Then clearly all edges from $p_2$ will end up in $q_2$. It is then evident that number of lost edges is at most $\pi$, since all lost edges have to be connected with $q_2$. 

\item The first edge is connected with $p_1$ and the last one is connected with either $p_1$ or $p_2$. This is the case that we will consider below. 

\item The first edge is connected with either $p_2$ or $p_3$ and the last one is connected with $p_3$. This is symmetric to the previous case.
\end{enumerate}

We now consider the fourth case. We will show that the number of lost edges that cannot be restored is at most $\pi$. Let $m_1$ be the number of edges that connect $p_1$ with $q_2$, and let $m_2 = \pi - z$ where $z$ is the last position in $p_2$ that is connected to a node in $q_2$.  In other words $m_2$ is the maximal range of positions at the end of $p_2$ that do not contain any node connected to $q_2$.   First, it is evident that all nodes in last $m_1$ positions of $p_1$ can only be connected to nodes in $q_2$; since we have suppressed $q_2$, all those positions in $p_1$ will no longer be connected with any edge.  We next consider the $m_3$ edges that connect $p_2$ to $q_2$. It is easy to see that $m_3\leq \pi-\mathtt{max}(m_1,m_2)$. We then consider the edges that connect $p_2$ to $q_3$. Those edges will all be suppressed. However, since the string $s$ is periodic and the last $m_1$ positions in $p_1$ are free, we can restore all edges that connect $q_3$ to the last $m_1$ positions in $p_2$ by connecting them to positions in $p_1$ instead. We thus cannot lose more than $\mathtt{max}(m_2-m_1,0)$ from this part. We now summarize that the edges we may lose: 
\begin{enumerate}
\item The edges that connect $p_1$ to $q_2$. This number is no more than 
$m_1$. 

\item The edges that connect $p_2$ to $q_2$. This number is no more than 
$\pi-\mathtt{max}(m_1,m_2)$. 

\item The edges that connect $p_2$ to $q_3$ and cannot be restored.  This number is at most $\mathtt{max}(m_2-m_1,0)$. 
\end{enumerate}
Thus the total number of edges that can be lost is at most 
\begin{equation}
\label{eq:a-3}
\pi-\mathtt{max}(m_1,m_2)+m_1+\mathtt{max}(m_2-m_1,0).
\end{equation}
Suppose now that $m_2\geq m_1$. Then (\ref{eq:a-3}) simplifies to $\pi-m_2+m_1+(m_2-m_1)=\pi$.  Otherwise if $m_2 < m_1$, the quantity simplifies to $\pi-m_1+m_1=\pi$.  We thus have proved the lemma for the simple version of edit distance where we do not have substitutions.

To extend the proof to the standard version of edit distance, we redefine the bipartite graph as follows.  We consider two types of edges. The first type is called \emph{matching} edges, and second type is called \emph{mismatching} edges. We require $s[i]\neq t[j]$ for a mismatching edge $(i,j)$. Let $z_1$ be the number of matching edges and $z_2$ be the number of mismatching edges. We keep the requirement that for any two edges $(i,j)$ and $(i',j')$ (regardless of their types), we have that $i\neq i'$, $j\neq j'$ and that $i<i'$ iff $j<j'$. It is easy to see that the standard edit distance is $k=m+n-2z_1-z_2$. We note that the proof for the simple version of edit distance still carries through, since we can prove that by deleting $p_2$ and $q_2$ we get a graph in which all but $\pi$ lost edges can be restored.  In particular, we can show in the fourth case at least $m_1$ edges can be restored regardless of their type.  Assuming the worst scenario in which all lost and non-restored edges are matching edges, we still have that the edit distance cannot increase since the term $2z_1+z_2$ cannot decrease by more than $2\pi$. This finishes the proof of the lemma. 
\end{proof}

\section{Missing Proofs in Sketching}

%

\subsection{Proof of Lemma~\ref{lem:recover}}
\label{proof:lem:recover}

\begin{proof}
First by the property of the random walk, for each block $u'$ in $s$ (or $t$) of size $B = 4\log n$, by a Chernoff bound its pre-image $u$ in $s$ (or $t$) must of size at least $3$ with probability at least $1 - 1/n^3$.  Then by a union bound on at most $O(n/\log n)$ blocks, with probability $1 - 1/n^2$ this holds for all blocks in $s$ and $t$, which we condition on in the rest of the proof.

Now define a mapping $f_s$ from the image string $s'$ (after the CGK embedding) back to the original string $s$, such that $f_s([i'..j']) = [i..j]$ iff character $s'(i')$ was copied from $s[i]$, and $s'(j')$ was copied from $s[j]$. Define $f_t$ the same way for $t'$ and $t$.


Consider the $n/B$ pairs of blocks of strings $s'$ and $t'$. Let $1 \le i_1 < \ldots < i_Z \le n/B$ denote the indices of those mismatched blocks.  Our goal is to recover $f_s([(i_z-1)B+1, i_z B])$ and $f_t([(i_z-1)B+1, i_z B])$ for all $z \in [Z]$.  We will show how to recover the position (starting and ending indices) of a particular block $f_s([(i_z-1)B+1, i_z B])$ in $s$.  We can do the same thing for $t$.  

Fix a $z \in [Z]$. Let $[x_1 .. y_1] = f_s([(i_z-1)B+1, i_z B])$.  Our goal is to compute $x_1$ and $y_1$.  Recall that each mismatched block in $s$ corresponds to a root-leaf path in the binary tree constructed for $s$ in the hierarchical structure, which can be identified by looking at the pairs of nodes with different hash signatures, which can be obtained from the differences between $P$ and $P'$.  From the length information stored in the leaf (level $1$) node of the path we can compute $\ell_1 = y_1 - x_1 + 1$.  Now let $[x_2 .. y_2] = f_s([(\lceil i_z/2 \rceil - 1)  2B + 1, \lceil i_z/2\rceil  2B])$ be the position of the substring corresponding to the level $2$ node of the path. We must have
\begin{eqnarray*}
\begin{array}{l}
(x_1, y_1)  =  \left\{
  \begin{array}{rl}
   (x_2, x_2 + \ell_1 - 1), & \text{if $i_z$ is odd},\\
   (y_2 - \ell_1 + 1, y_2), & \text{if $i_z$ is even}.
  \end{array}
  \right.
\end{array}
\end{eqnarray*} 
Thus once $(x_2, y_2)$ is obtained, we can also obtain $(x_1, y_1)$.  We thus can compute $(x_1, y_1)$ in a recursive way, and the recursion will finally reach the root where we have $[x_L .. y_L] = [1..n]$ (thus $x_L = 1$ and $y_L = n$), where $L = \log(3n/B)$ is the height of the tree. 

Once we have the positions of all the mismatched blocks, we can use the differences between $Q$ and $Q'$ to recover their contents. 

Now we can enumerate the edges of the corresponding effective alignment $\A$ in a backward greedy fashion.  Consider for a $z \in [Z]$ with $i_z + 1 < i_{z+1}$ the gap between the $i_z$-th pair of mismatched blocks and the $i_{z+1}$-th pair of mismatched blocks in $s'$ and $t'$. Let $p = i_zB + 1$ and $q = (i_{z+1} - 1)B$.  We thus have $s'[p..q] = t'[p..q]$.   Let $u = f_s(q)$ and $v = f_t(q)$.   Since  $s'[q]=t'[q]$, we must have $s[u] = t[v]$ and the random walk must take the same action when reading $s[u]$ and $t[v]$, that is, either the next walk state is $(u,v)$ or $(u+1, v+1)$.  Now we also have $s'[q-1]=t'[q-1], \ldots, s'[p] = t'[p]$, we can thus apply this argument backwards, and get $s[u-1] = t[v-1], s[u-2] = t[v-2], \ldots$, which means we can add edges $(u-1, v-1), (u-2, v-2), \ldots$ to $\A$.  We can continue this process until we reach the edge $(u-\beta, v-\beta)$ where either $u-\beta = f_s(p-1) + 1$ or $v-\beta = f_t(p-1) + 1$, at which point we have three cases:
\begin{enumerate}
\item We have both $u-\beta = f_s(p-1) + 1$ and $v-\beta = f_t(p-1) + 1$.  In this case we do nothing.

\item We have $u-\beta = f_s(p-1) + 1$ but $v-\beta = f_t(p-1) + 2$.  In this case we make $t[v - \beta -1]$ a singleton.

\item We have $v-\beta = f_t(p-1) + 1$  but $u-\beta = f_s(p-1) + 2$.  In this case we make $s[v - \beta -1]$ a singleton.
\end{enumerate}
The construction of $\A$ completes when all the matching edges are added.  In fact, we do not even need to enumerate edges one by one but just compute the first and last edges of these clusters of consecutive edges.
\end{proof}

\subsection{Proof of Lemma~\ref{lem:decode}}
\label{appendix:proof_decode}
We first show how to compute $\I$ using a $\rho$-way merge.  Recall that each effective alignment $\A$ can be represented by a set of clusters in the form of $\C_\ell = (u_\ell, v_\ell, \eta_\ell)$ where $(u_\ell, v_\ell)$ is the first edge of the $\ell$-th cluster in $\A$ and $\eta_\ell$ is the number of consecutive edges in the cluster, plus at most $N$ singletons in between  (recall that $N = \poly(K \log n)$ is a parameter we introduce for the definition of a good random walk). The number of clusters is clearly upper bounded by the number of singletons (plus $1$).  For each effective alignment $\A_j$, let $G_j$ be the set of nodes consisting of all the singletons and all the boundary (first and last) nodes of the clusters in  $s$.  Let $G = \cup_{j \in [\rho]} G_j$.  Thus $\abs{G} = \poly(K \log n)$.  We now scan the nodes in $G$ sequentially from left to right, during which we maintain a binary search tree $\T$ to keep track of all currently ``active'' clusters, that is, when encountering a node $G$ that is the first node of a cluster $\C = (u, v, \eta)$,  we add $\C$ to $\T$ with the key $(u - v)$ (i.e., the shift of the edges in $\C$), and when encountering the last node of a cluster $\C$, we simple remove $\C$ from $\T$.  At any step during the scan, if $\T$ has $\rho$ nodes all of which have the same key (shift), then we add the corresponding edge (determined by the node in $s$ and the shift) to $\I$. In this way the total time for computing $\I$ is bounded by $\poly(K \log n)$.

Since all the edges in $\I$ will be included in the optimal alignment,  we just need to align the remaining nodes in the gaps between the clusters in $\I$ in the optimal way. We will make use of the algorithms in \cite{LMS98} which involves $O(K^2)$ longest common prefix (\LCP) queries (we refer the readers to Section~\ref{sec:sim-stream} for the idea of the algorithm in \cite{LMS98}, but presented in a space-efficient manner).  We will show that we can construct a data structure $\D$ of size $\poly(K \log n)$ to answer the $O(K^2)$ \LCP\ queries in $\poly(K \log n)$ time, by making use of the fact that most of the remaining nodes (except $\poly(K \log n)$ ones) form at most $\poly(K \log n)$ periodic substrings with periods of lengths at most $\poly(K \log n)$.

We partition $[1..n]$ into two types of alternative regions, called {\em short} regions and {\em long} regions respectively, as follows. We view the nodes in $G$ naturally partitioning $[1..n]$ into a set of intervals, and call an interval {\em big} if its length is larger than $6N$, and {\em small} otherwise.  Note that a big interval is always followed by a small interval  because there must be at least one singleton after the big interval. We now scan the nodes in $G$ from left to right. The first (short) region is the concatenation of all small intervals plus the first $2N$ coordinates of the first big interval  (denoted by $[x..y]$) encountered; the second (long) region consists of the coordinates $[x+N .. y-N]$ inside the big interval $[x..y]$; the third (short) region consists of a set of small consecutive intervals plus the $2N$ coordinates of the preceding and succeeding big intervals; and so on.  Note that each pair of adjacent short and long regions overlap on $N$ nodes.

The observation is that for each long region $[x + N .. y - N]$, both $s[x + N .. y - N]$ and $t[x + N .. y - N]$ are periodic with the same period of length $\pi \le 2N$.  To see this, note that each node $s[i]\ (i \in [x,y])$ is connected by at least two edges $(i, i')$ and $(i, i'')\ (i' \neq i'')$ in two effective alignments such that $\abs{i - i'} \le N$ and $\abs{i - i''} \le N$, which implies $\abs{i' - i''} \le 2N$ and $t[i'] = t[i'']$, and consequently $t[x+N .. y-N]$ is periodic with period length at most $2N$.  Applying a similar argument we can show that $s[x+N.. y-N]$ is periodic with periods of the same length.  

For each short region $[i..j]$, we build a data structure $\ST$ consisting of a suffix tree on $s[i..j]$ and $t[i..j]$, which is used for answering \LCP\ queries within the region.  For each long region $[i..j]$, we build a data structure $\PE$ consisting of the period length $\pi$ of the region and a suffix tree on $s[i + N .. i + N + 2\pi - 1]$ and $t[j + N .. j + N + 2\pi - 1]$.  We can answer an \LCP\ query $(\alpha, \beta)$ within each long region as follows: we first query the suffix tree in $\PE$ with $((\alpha - 1) \bmod \pi + 1, (\beta - 1) \bmod \pi + 1)$. If the answer is less than $\pi$, then the answer is immediately returned; otherwise if the answer is in the range $[\pi, 2\pi]$, since $s[i+N .. j-N]$ and $t[i+N .. j-N]$ are periodic, we know that the \LCP\ starting from $(\alpha, \beta)$ can be further extended, in which case the answer is $r \pi + \text{\LCP}(\alpha + r\pi, \beta + r\pi)$, where $r$ is the smallest integer such that both $\alpha + r\pi$ and $\beta + r\pi$ fall into the range $[j - 2N + 1 .. j]$, and $\text{\LCP}(\alpha + r\pi, \beta + r\pi)$ is the answer to the \LCP\ query on $(\alpha + r\pi, \beta + r\pi)$ for which we will recursively query the $\ST$ structure of the next short region.

The whole data structure can be represented as $\D = (\ST_1, \PE_1, \ST_2, \PE_2, \ldots, \ST_Z)$.  Give a query $(\alpha, \beta)$ with $\abs{\alpha - \beta} \le K$, we scan the sequence of regions to locate the first region $[i..j]$ such that $i \le \alpha, \beta \le j$.  Note that we can always find such a region since the adjacent regions overlap on at least $N \ge K$ nodes.  Let $\ST_z$ or $\PE_z$ $(z \in [Z])$ be the associated data structure of that region, using which we can get an answer $h$ for $\text{LCP}(\alpha, \beta)$ {\em within} the region such that $s[\alpha .. \alpha + h - 1] = t[\beta .. \beta + h - 1]$. We now have two cases.
\begin{enumerate}
\item We have $\alpha + h - 1 = j$ or $\beta + h - 1 = j$.  In this case, $\text{\LCP}(\alpha, \beta) = h + \text{LCP}(\alpha + h, \beta + h)$, and we can compute $\text{LCP}(\alpha + h, \beta + h)$ by recursively querying $\PE_z$ or $\ST_{z+1}$.  

\item We have $\alpha + h - 1 < j$ and $\beta + h - 1 < j$. In this case we simply return $h$.
\end{enumerate}

It is easy to see that the size of $\D$ can be bounded by $\poly(K \log n)$, and can be constructed in time $\poly(K \log n)$.  An \LCP\ query can be answered in $O(N)$ time, and thus $O(K^2)$ queries can be answered in time $O(K^2 N)$.  We conclude that both the space and time used in the decoding phase are both upper bounded by $\poly(K \log n)$.